\documentclass[11pt]{article}

\usepackage{latexsym}
\usepackage{amsmath}
\usepackage{amstext}
\usepackage{amsfonts}
\usepackage{amsopn}

\usepackage{ifthen}

\newtheorem{theorem}{Theorem}[section]
\newtheorem{lemma}[theorem]{Lemma}
\newtheorem{corollary}[theorem]{Corollary}
\newtheorem{proposition}[theorem]{Proposition}
\newtheorem{definition}{Definition}[section]

\newenvironment{ProofDummyEnv}{}{}

\newenvironment{proof}[1][none]{\begin{proofby}[#1]{}}{\end{proofby}}

\newenvironment{proofby}[2][none]{\par\noindent{\bf Proof:} #2
\renewenvironment{ProofDummyEnv}{\begin{#1}}{\end{#1}}%
\begin{ProofDummyEnv}}%
{\QED\end{ProofDummyEnv}}

\newcommand{\QED}{\nopagebreak\hfill $\Box$}

\newcommand{\prob}[2][]{\text{\bf Pr}\ifthenelse{\not\equal{}{#1}}{_{#1}}{}\!\left[#2\right]}
\newcommand{\expect}[2][]{\text{\bf E}\ifthenelse{\not\equal{}{#1}}{_{#1}}{}\!\left[#2\right]}
\newcommand{\given}{\ \mid\ }

\newcommand{\abs}[1]{\left| #1 \right|}

\usepackage{fullpage}
\usepackage{amsmath,amsfonts}
\usepackage{enumitem}
\usepackage{graphicx}
\usepackage{pst-all}

\def\eps{\epsilon}

\def\cali{{\cal I}}

\newtheorem{thm}{Theorem}[section]
\newtheorem{cor}[thm]{Corollary}
\newtheorem{lem}[thm]{Lemma}

\newtheorem{claim}[thm]{Claim}

\DeclareMathOperator{\OPT}{OPT}

\newcommand{\alg}{{\cal A}}
\newcommand{\mech}{{\cal M}}

\newcommand{\ialg}{\bar{\alg}}

\newcommand{\val}{v}
\newcommand{\vals}{{\mathbf \val}}
\newcommand{\valsmi}{\vals_{-i}}
\newcommand{\vali}[1][i]{{\val_{#1}}}

\newcommand{\ival}{\bar{\val}}
\newcommand{\ivals}{\bar{\vals}}

\newcommand{\ivali}[1][i]{{\ival_{#1}}}

\newcommand{\dist}{F}
\newcommand{\dists}{{\mathbf \dist}}
\newcommand{\distsmi}{\dists_{-i}}
\newcommand{\disti}[1][i]{{\dist_{#1}}}

\newcommand{\dens}{f}

\newcommand{\densi}[1][i]{{\dens_{#1}}}

\newcommand{\price}{p}
\newcommand{\prices}{{\mathbf \price}}
\newcommand{\pricei}[1][i]{{\price_{#1}}}

\newcommand{\alloc}{x}
\newcommand{\allocs}{{\mathbf \alloc}}

\newcommand{\alloci}[1][i]{{\alloc_{#1}}}

\newcommand{\est}{y}
\newcommand{\ests}{{\mathbf \est}}

\newcommand{\esti}[1][i]{{\est_{#1}}}

\newcommand{\discalloc}{\dot{\alloc}}
\newcommand{\discallocs}{\dot{\allocs}}

\newcommand{\discalloci}[1][i]{{\discalloc_{#1}}}

\newcommand{\ialloc}{\bar{x}}

\newcommand{\ialloci}[1][i]{{\ialloc_{#1}}}

\newcommand{\feasibles}{{\cal X}}

\newcommand{\Gbar}{\bar{G}}
\newcommand{\gbar}{\bar{g}}

\newcommand{\intset}{{\cal I}}
\newcommand{\intseti}[1][i]{\intset_{#1}}
\newcommand{\intsets}{\intset\!\!\!\intset}

\newcommand{\cost}{c}

\newcommand{\approxratio}{\beta}

\newcommand{\stair}[1]{\text{\sc Stair}\left(#1\right)}

\newcommand{\monoints}[1]{\text{\sc MonoInts}\left(#1\right)}
\newcommand{\resampalg}[2][\intsets]{\text{\sc Resample}\left(#2, #1\right)}
\newcommand{\discalg}[2][\eps]{\text{\sc Disc}_{#1}\left(#2\right)}
\newcommand{\statalg}[2][\eps]{\text{\sc Iron}_{#1}\left(#2\right)}
\newcommand{\corralg}[2][\eps]{\text{\sc Mono}_{#1}\left(#2\right)}
\newcommand{\combalg}[2][\eps]{\text{\sc Comb}_{#1}\left(#2\right)}
\newcommand{\trimalg}[2][\eps]{\text{\sc Trim}_{#1}\left(#2\right)}

\newcommand{\corralloc}{\hat{x}}
\newcommand{\corralloci}[1][i]{{\corralloc_{#1}}}

\newcommand{\mumax}{\mu_{\max}}

\newcommand{\softO}{\tilde{O}}

\newcommand{\stairfrac}{\delta}

\begin{document}

\title{Bayesian Algorithmic Mechanism Design}

\iffalse
\numberofauthors{2}

\author{
\alignauthor
Jason D. Hartline\titlenote{Supported in part by NSF Grant CCF-0830773 and NSF Career Award CCF-0846113.}\\
       \affaddr{Northwestern University}\\
       \affaddr{Evanston, IL, USA}\\
       \email{hartline@eecs.northwestern.edu}
\alignauthor
Brendan Lucier\\
       \affaddr{University of Toronto}\\
       \affaddr{Toronto, ON, Canada}\\
       \email{blucier@cs.toronto.edu}
}
\fi

\author{
Jason D. Hartline\thanks{Supported in part by NSF Grant CCF-0830773 and NSF Career Award CCF-0846113.} \thanks{Northwestern University, Evanston, IL, USA.  email: \texttt{hartline@eecs.northwestern.edu}}
\and
Brendan Lucier\thanks{University of Toronto, Toronto, ON, Canada.  email: \texttt{blucier@cs.toronto.edu}}
}

\maketitle

\begin{abstract}
The principal problem in algorithmic mechanism design is in merging
the incentive constraints imposed by selfish behavior with the
algorithmic constraints imposed by computational intractability.  This
field is motivated by the observation that the preeminent approach for
designing incentive compatible mechanisms, namely that of Vickrey,
Clarke, and Groves; and the central approach for circumventing
computational obstacles, that of approximation algorithms, are
fundamentally incompatible: natural applications of the VCG approach
to an approximation algorithm fails to yield an incentive compatible
mechanism.  We consider relaxing the desideratum of (ex post)
incentive compatibility (IC) to Bayesian incentive compatibility
(BIC), where truthtelling is a Bayes-Nash equilibrium (the standard
notion of incentive compatibility in economics).  For welfare
maximization in single-parameter agent settings, we give a general
black-box reduction that turns any approximation algorithm into a
Bayesian incentive compatible mechanism with essentially the
same\footnote{More specifically, we obtain a polynomial time 
  approximation scheme with an $\epsilon$ loss that is either 
  additive or multiplicative, depending on the problem setting. 
  This error term arises from statistical methods that seem
  necessary for a black-box reduction.} approximation factor.  
%In
%other words, for a large relevant class of problems and slight
%relaxation of the standard (in computer science) solution concept, we
%solve the central problem in algorithm mechanism design.
\end{abstract}

%\begin{document}
%\setcounter{pqage}{0}

%\begin{titlepage}

%\category{J.4}{Social And Behavioral Sciences}{Economics}

%\terms{Algorithms, Economics, Theory}

%\keywords{Mechanism design, Bayesian incentive compatibility, algorithms, social welfare.}

% remove page number.
%\renewcommand{\thepage}{}

%\end{titlepage}

%\newpage

\section{Introduction}

%%
%% reduction
%%
{\em Can any approximation algorithm be converted into an
  approximation mechanism for selfish agents?}  This question is
framed by a fundamental incompatibility between the standard economic
approach for the design of mechanisms for selfish agents (the {\em
  Vickrey-Clarke-Groves} (VCG) mechanism) and the standard algorithmic
approach for circumventing computational intractability (approximation
algorithms).  The conclusion from this incompatibility, driving much
of the field of algorithmic mechanism design, is that incentive and
algorithmic constraints must be dealt with simultaneously (See e.g.,
\cite{lav-07}).  For a large, important class of problems, we arrive
at the opposite conclusion: {\em there is a general
  approximation-preserving reduction from mechanism design to
  algorithm design!}

%%
%% BIC
%%
The goal of mechanism design is to construct the rules for a system of
agents so that in the equilibrium of selfish agent behavior a desired
objective is obtained.  For settings of incomplete information the
standard game theoretic equilibrium concept is {\em Bayes-Nash
  equilibrium} (BNE), which is defined by mutual best response when
the {\em prior distribution} of agent payoffs is {\em common knowledge}.  The
{\em revelation principle}~\cite{mye-81} suggests that when looking
for mechanisms with desirable Bayes-Nash equilibria, one must look no
further than those with truthtelling as a Bayes-Nash equilibrium, also
known as {\em Bayesian incentive compatible} (BIC) mechanisms.  Almost
all of the computer science literature has focused on the BIC subclass
of {\em ex post incentive compatible} (IC) mechanisms where
truthtelling is a {\em dominant strategy}.  While IC is aesthetically
appealing because it is congruous with worst-case-style results, it
is not generally without loss!  

%%
%% motivation for BIC
%%
This loss is evident in the computer science theory of IC mechanism
design, which is described most characteristically by impossibility.  For
instance, for single-minded combinatorial auctions of $m$ items, the
optimal worst-case approximation factor (under standard complexity
theoretic assumptions) is $\sqrt m$~\cite{LOS-99}.  With such
strong lower bounds, a relevant theory must make relaxations.  For many
problems within the realm of computer systems; e.g., online auctions
(eBay), advertising auctions (Google, Yahoo!, MSN), file sharing
(BitTorrent), routing (TCP/IP), scheduling (SETI@home), and video
streaming (YouTube); high volume should enable demand distributions to
be estimated.  With demand distributions, the natural algorithmic and
mechanism design problems are Bayesian.

%%
%% single-dimensional and monotonicity
%%
Agent incentives in Bayesian mechanism design are very well understood
in single-parameter settings, where each agent has a single
independent private value for receiving a service (see,
e.g,~\cite{mye-81}).  For the single parameter setting it is known
that a mechanism is BIC if and only if (a) the probability an agent is
served (a.k.a.~the {\em allocation rule}) is monotone non-decreasing
in the agent's value for service, and (b) the agent's expected payment
(a.k.a.~the {\em payment rule}) is of a particular form 
identified precisely from the
allocation rule.\footnote{Probabilities and expectations above are
  taken with respect to both the distribution of agent values and
  possible randomization in the mechanism.}

%
% monotonizing rules.
%
The main challenge in reducing BIC (or IC) mechanism design to
algorithm design is that approximation algorithms do not generally
have monotone allocation rules.  Our reduction shows that in a
Bayesian setting we can convert any non-monotone allocation rule into
a monotone one without compromising its social welfare.  The main
technical observation that enables this reduction is that, in a
Bayesian setting, we can focus on a single agent for whom the
allocation rule is not monotone, apply a transformation that fixes
this non-monotonicity (and weakly improves our objective), and {\em no
  other agents are affected} (in a Bayesian sense).  Therefore, we can
apply the transformation independently to each agent.
%
%
% reduction
%
Our reduction is as follows:
\begin{enumerate}
\item For each agent, identify intervals in which the agent's
  allocation rule is non-monotone.  (This is a property of the
  distribution and algorithm and can be done prior to considering 
  any agent bids.)
\item For each agent, if their bid falls in an (above identified)
  interval, redraw the agent's bid from the prior distribution conditioned
  on being within the interval.
\item Run the approximation algorithm on the resulting bids and
  output its solution.
\end{enumerate}
Notice that under the assumption that the original values are drawn
according to the common prior, the redrawing of values does not alter
this prior.

%%
%% ironing, sampling, prices
%%
Three items must be clarified.  First, there are many ways one might
try to choose intervals in Step 1 of the reduction and most of them
are incompatible with mechanism design.  To address this issue, we
develop a monotonizing technique for allocation rules (adapted from
the standard {\em ironing procedure} from the field of Bayesian
optimal mechanism design~\cite{mye-81}).  Second, we are unlikely to
have access the functional form of the allocation rule.  To address
this issue, we estimate the allocation rule by sampling the distribution and
making black-box calls to the algorithm.  These estimates can be made
precise enough to enable arbitrary small loss in welfare (i.e., a
fully polynomial time approximation scheme).  Finally, we must
also determine payments for our monotonized allocation rule.  For
this, a general approach of Archer et al.~\cite{APTT-03} suffices.

%%
%% objective and feasibility
%%
%Our results apply generally to single-parameter agent settings where
%the designer's objective is to maximize the social welfare (e.g.,
%single-minded combinatorial auctions).
%Agents each have a private value for service.  The designer either has
%a feasibility constraint on the set of agents that can be
%simultaneously served or a cost function over the set of served
%agents.  The social welfare is the sum of the values of the agents
%served less the designer's cost.
Our results apply generally to single-parameter agent settings where
the designer's objective is to maximize the social welfare (e.g.,
single-minded combinatorial auctions).  In the most general form, such
an algorithmic problem can be written as finding an allocation $\allocs
= (\alloci[1],\ldots,\alloci[n])$ to maximize $\sum_i \vali \alloci -
\cost(\allocs)$ for agent valuations $\vals =
(\vali[1],\ldots,\vali[n])$ and cost function $\cost(\cdot)$.  For
instance, the {\em multicast auction} problem of Feigenbaum et
al.~\cite{FPS-00} is the special case where the $\cost(\allocs)$ is
the sum of the costs of tree edges necessary to connect all agents
served by $\allocs$ to the root (generally, the Steiner tree problem).
A special and relevant case occurs when costs are zero for $\allocs$ in
some feasible set system $\feasibles$ and all other allocations are
infeasible (i.e, $\cost(\allocs) = 0$ if $\allocs \in \feasibles$ and
$\infty$ otherwise).  For the single-item auction, $\feasibles$ is the
collection of all sets of cardinality at most one; and for single-minded combinatorial
auctions, $\feasibles$ contains all sets of agents with non-intersecting
desired bundles.  Our most general result does not need any
restrictions on the cost function or the set system.  In particular,
costs can be arbitrarily non-monotone or the set system
non-downward-closed (e.g. public good or scheduling problems).

%%
%% results
%% 
For any $\epsilon$, we give a black box reduction that, in polynomial
time in the number of agents and $1/\epsilon$, converts any
approximation algorithm into a BIC
mechanism with an additive loss of $\epsilon$ to the social welfare.  We also give a pseudo-polynomial time reduction to a BIC mechanism with a multiplicative loss of $\epsilon$, and a fully polynomial time approximation scheme for the special case of downward-closed feasibility problems.
Thus, the
approximation complexity of social welfare in single-parameter
settings is the same for algorithms and BIC mechanisms.

%%
%% application of our results.
%%
For the most studied single-parameter mechanism design problems, the
performance of the best ex post IC approximation mechanism matches the
best approximation algorithm (e.g., single-minded combinatorial
auctions~\cite{LOS-99} and related machine scheduling~\cite{DDDR-08}).
None-the-less, our approach gives the best known BIC approximation
mechanism for many problems, such as auctions under various graph
constraints~\cite{AADK-02} and auctions of convex
bundles~\cite{BB-04}.

%%
%% conclusion
%%
Our result demonstrates that there is no gap between algorithmic approximation and
approximation by BIC mechanisms.  The remaining (theoretical)
question of gaps in approximation factors imposed by
incentive constraints is thus focused on whether BIC is more powerful
than IC for social welfare maximization.  For other
non-welfare-maximization objectives (e.g., makespan) the question of a
general reduction remains open.

\paragraph*{Related Work}

The design of ex post IC mechanisms for social welfare problems is
well studied, notably for the specific settings of combinatorial
auctions \cite{APTT-03, BB-04, LS-05, LOS-99}.  Lehmann et
al.~\cite{LOS-99} introduced the problem of polynomial time
approximation of social welfare for single-minded combinatorial
auctions and give a mechanism that matches the best algorithmic
approximation factor.  Archer et al.~\cite{APTT-03} considered the
setting where there are many (at least logarithmic) copies of each
item and gave a $(1+\epsilon)$-approximation mechanism.  Archer and
Tardos~\cite{AT-01} gave a (single-parameter) related machine
scheduling mechanism that approximates the makespan.  Dhangwatnotai et
al.~\cite{DDDR-08} gave a mechanism for related machine sched\-ul\-ing
that approximates makespan and matches the algorithmic lower bound.
All of the above results are for ex post incentive compatible
mechanisms.

There has been a large literature on multi-parameter combinatorial
auctions and approximation, but this is only tangentially related to
our work so we do not cite it exhaustively.

The literature contains a few reductions from mechanism design to 
algorithm design of varying degrees of generality.
Lavi and Swamy \cite{LS-05} consider IC mechanisms for multi-parameter
packing problems and give a technique for constructing a
(randomized) $\approxratio$-approximation mechanism from any
$\approxratio$-approximation algorithm that verifies an integrality
gap.  Babaioff et al.~\cite{BLP-09} look at the equilibrium notion of
{\em algorithmic implementation in undominated strategies} and gives a
technique for turning a $\approxratio$-algorithm into a
$\approxratio (\log \val_{max})$-approximation mechanism.  This
solution concept requires that no agent plays a strategy that is
dominated by an easy to find strategy.  Their approach applies to
single-valued combinatorial auctions and does not require the
mechanism to know which bundles each agent desires.

There have been a few related studies of Bayes-Nash equilibrium.
Christodoulou et al.~\cite{CKS-08} consider Bayes-Nash equilibria of
simultaneous Vickrey auctions in a combinatorial setting and show that
these give a 2-approximation when agents' valuations are submodular.
Gairing et al.~\cite{GMT-05} consider Bayes-Nash equilibria of a
routing game and study worst-case performance.
Borodin and Lucier~\cite{BL-10} study worst-case performance of Bayes-Nash equilibria in combinatorial mechanisms based on greedy algorithms.

There are many papers on profit maximization that consider
Bayesian design settings.  These papers do not tend to consider
computational constraints and for many of these (non-computational)
settings the restriction to ex post incentive compatibility is without
loss.%
\footnote{One non-computational setting where ex post incentive
  compatibility is with loss is when the profit maximizing seller has
  a strict no-deficit constraint \cite{CHRR-06}.}  One notable
exception is Bhattacharya et al.~\cite{BGGM-09} which focuses on the
problem of selling heterogeneous goods to agents with linear
valuations.  They construct a polynomial time $4$-approximation
mechanism.  Their approximation result requires that the type
distributions satisfy the monotone hazard rate assumption.  In a
spirit similar to this paper, they make heavy use of the Bayesian
setting to obtain a polynomial runtime.

\paragraph*{Organization} 
We describe in detail the model for single-parameter agents, Bayesian
approximation, Bayesian incentive compatibility, and foundational
economic theory in Section~\ref{sec:model}.  In
Section~\ref{sec:ideal} we give the reduction in an ideal setting
where the allocation rule of the algorithm for the given distribution
on agent values is precisely known.  This reduction is lossless.  In
Section~\ref{sec:black-box} we develop the reduction in the black-box
model where we must sample the distribution and run the algorithm to
determine its allocation.  Conclusions and open problems are discussed
in Section~\ref{sec:conclusions}.

\section{Model and Definitions}

\label{sec:model}

\paragraph*{Algorithms}
%
% single-parameter agents
%
We consider algorithms for binary single-pa\-ram\-e\-ter agent settings.  An
algorithm in such a setting must select a set of agents to serve.
This {\em allocation} is denoted by $\allocs =
(\alloci[1],\ldots,\alloci[n])$ where $\alloci$ is an indicator for
whether or not agent $i$ is served.  Agent $i$ has {\em valuation}
$\vali$ for being served.  Without loss for non-negative
bounded-support distributions, we will assume $\vali \in
[0,1]$.\footnote{The bounded support assumption is unnecessary except
  for our results using sampling, where we believe it is realistic.}
The vector $\vals = (\vali[1],\ldots,\vali[n])$ of valuations is the
{\em valuation profile}.  

In the {\em general costs setting}, the seller may have some cost
function $\cost(\cdot)$ over allocations representing the cost of
serving the allocated set (e.g., Steiner tree problems \cite{FPS-00}).
The {\em general feasibility setting} is the special case where 
costs are zero (feasible) or infinity (infeasible). 
%General
%feasibility problems 
These
include scheduling and public good problems.  An
important subclass are 
%{\em downward-closed (feasibility) settings}
{\em downward-closed settings}
where any subset of a feasible set is feasible.  Downward closed
settings include single-minded combinatorial auctions~\cite{LOS-99}
and knapsack auctions~\cite{AH-06}.

%
% algorithms and objective
%
An algorithm $\alg$ is simply an {\em allocation rule} that maps
valuation profiles to allocations.  The allocation rule for $\alg$ we
will denote by $\allocs(\vals)$.  Our objective is the {\em social
  welfare} which is $\alg(\vals) = \sum_i \vali \alloci(\vals) -
\cost(\allocs(\vals))$.  We allow $\alg$ to be randomized in which
case $\alloci(\vals)$ is a random variable; $\alg(\vals)$ denotes the
expected welfare of the algorithm for valuation profile $\vals$.
$\OPT(\vals)$ will denote the maximum social welfare.

%
% Bayesian
%
We will consider these algorithmic problems in a Bayesian
(a.k.a.~stochastic) setting where the valuations of the agents are
drawn from a product distribution $\dists = \disti[1] \times \cdots
\times \disti[n]$.  Agent $i$'s {\em cumulative distribution} and {\em
  density} functions are denoted $\disti$ and $\densi$, respectively.
The distribution is assumed to be {\em common knowledge} to the
agents and designer.  
% Note: for now these definitions are being kept here, as we state the
% main theorem using them (even though we prove a weaker version that
% does not require them).
%Some of our bounds will be based on
%$\val_{max}$, an upper bound on any agent's valuation, and $\mumax =
%\max_i\expect{\vali}$, the maximum expected valuation of any agent.

%%
%% worst-case and Bayesian approximation.
%%
The pair $(\cost(\cdot),\dists)$ defines a setting for
single-parameter algorithm design which we will take as implicit.  For
this setting, the optimal expected welfare is $\OPT = \expect[\vals
  \sim \dists]{\OPT(\vals)}$ and the algorithm's expected welfare is
$\alg = \expect[\vals \sim \dists]{\alg(\vals)}$.  An algorithm
is a {\em worst-case $\approxratio$-approximation} if for all $\vals$,
$\alg(\vals) \geq \OPT(\vals)/\approxratio$.  An algorithm is a {\em Bayesian
  $\approxratio$-approx\-imation} if $\alg \geq \OPT/\approxratio$.

\paragraph*{Mechanisms}
%%
%% mechanism design
%%
A mechanism $\mech$ consists of an {\em allocation rule} and a {\em
  payment rule}.  We denote by $\allocs(\vals)$ and $\prices(\vals)$
the allocation and payment rule of an implicit mechanism $\mech$.  We
assume agents are {\em risk neutral} and individually desire
to maximize their expected {\em utilities}.  Agent $i$'s utility for
allocation $\allocs$ and payments $\prices$ is $\vali\alloci -
\pricei$.  We consider single-round, sealed-bid mechanisms where
agents simultaneously bid and the mechanism then computes the
allocation and payments.

%%
%% IC vs BIC
%%
Our goal is a mechanism that has good social welfare in
equilibrium.  The standard economic notion of equilibrium for games of
incomplete information is {\em Bayes-Nash equilibrium} (BNE).  The
revelation principle says that any equilibrium that is implementable in
BNE is implementable with truthtelling as the BNE strategies of the
agents.\footnote{The revelation principle holds even in computational
  settings; any BNE for which the agent strategies and the mechanism
  can be computed in polynomial time can be converted into a
  polynomial time BIC mechanism.}  Meaning: an agent that believes 
%that
the other agents are reporting their values truthfully as given by the
distribution has a best response of also reporting truthfully.  A
mechanism with truthtelling as a BNE is {\em Bayesian incentive
  compatible} (BIC).\footnote{Much of the computer science literature
  on mechanism design focuses on dominant strategy equilibrium (DSE)
  and {\em ex post incentive compatibility} (IC).  This is not without
  loss in many settings and therefore should be considered with care
  when addressing computational questions in mechanism design.}

%%
%% ex post IR
%%
%All of our mechanisms will be ex post {\em individual rational} (IR).
%Meaning: any truthful reporting agent will obtain non-negative utility
%in worst-case.  

%%
%% characterization
%%
It will be useful to consider agent $i$'s expected payment and
probability of allocation conditioned on their value.  To this end, denote
$\pricei(\vali) = \expect[\vals,\alg]{\pricei(\vals) \given \vali}$
and $\alloci(\vali) = \expect[\vals,\alg]{\alloci(\vals) \given \vali}$.  The following theorem characterizes BIC mechanisms.
\begin{theorem} \cite{mye-81} \label{t:bic}
A mechanism is BIC if and only if for all agents $i$:
\begin{itemize}
\item $\alloci(\vali)$ is monotone non-decreasing, and
\item $\pricei(\vali) = \vali \alloci(\vali) - \int_0^{\vali}
  \alloci(z) dz + \pricei(0)$.
\end{itemize}
Usually, $\pricei(0)$ is assumed to be zero.
\end{theorem}
This motivates the following definition:
\begin{definition}
An allocation rule $\allocs(\cdot)$ is {\em monotone} for distribution
$\dists$ if $\alloci(\vali)$ is monotone non-decreasing for all $i$.
An algorithm is {\em monotone} if its allocation rule is monotone.  
\end{definition}
From Theorem~\ref{t:bic}, BIC and monotone are equivalent and we will
use them interchangeably for both algorithms and mechanisms, though
we will prefer ``BIC'' when the focus is incentive properties and
``monotone'' when the focus is algorithmic properties.

\paragraph*{Computation}
%%
%% computational model
%%
Our main task in demonstrating that the approximation complexity of
algorithms and BIC mechanisms is the same by giving an
approximation-preserving reduction from the BIC mechanism design
problem to the algorithm design problem.  In other words, we use
the algorithm's allocation rule to compute the mechanism's allocation
and payment rules.  As we are in a Bayesian setting this computation
will also need access to the distribution.  

We consider two models of
computation: an {\em ideal model} and a {\em black-box model}.
In the ideal model, we will assume we have explicit access to the
functional form of the distribution and allocation rule and we will
assume we can perform calculus on these functions.  While this model
is not realistic, we present it for the sake of clarity in explaining
the economic theory that drives our results.  In the black-box
model we will assume we can query the algorithm on
any input and that we can sample from the distribution on any
subinterval of its support.  Our philosophy is that the ideal model is
predictive of what is implementable in polynomial time and we verify
this philosophy by instantiating approximately the same reduction
under the black-box model.

%% We give our algorithm runtimes in the black-box model in terms of the
%% number of calls to the algorithm black-box.  This is justified as each
%% such call requires at most a linear number of calls to the
%% distribution black-box and no significant additional computation is
%% needed.

\section{Reduction: Ideal Model}
\label{sec:ideal}

In this section we prove that, in the ideal model, any Bayesian 
algorithm can be made BIC without loss of performance.

\begin{theorem}
\label{t:ideal}
In the ideal model and general cost settings, a BIC algorithm $\ialg$ can be computed from any algorithm $\alg$.  Its expected social welfare satisfies $\ialg \geq \alg$.
\end{theorem}

Theorem~\ref{t:ideal} implies an immediate corollary for Bayesian
approximation. 

\begin{corollary}
\label{c:ideal}
In the ideal model and general cost settings, a BIC Bayesian
$\approxratio$-approximation $\ialg$ can be computed from
any Bayesian $\approxratio$-approximation algorithm, $\alg$.
\end{corollary}

Corollary \ref{c:ideal} applies in the special case that $\alg$ 
is a worst-case $\approxratio$-approximation, but the resulting
BIC algorithm $\ialg$ will not necessarily be a worst-case
$\approxratio$-approximation.  See Appendix~\ref{app:Bayesian-approx} 
for a concrete example.

Let us build some intuition for the requirements of Theorem \ref{t:ideal}. 
Suppose that we are given an algorithm $\alg$ that is monotone for the
distribution $\dists$.  Then $\alg$ is already BIC and specifies the
allocation rule $\allocs(\cdot)$, so we must only compute the payment
rule.  In our ideal model this is trivial given the formula from
Theorem~\ref{t:bic}.

Now suppose we have a non-monotone Bayesian $\approxratio$-approx\-imation
algorithm $\alg$ with allocation rule $\allocs(\cdot)$.
%\footnote{
%  We note that Theorem \ref{t:ideal} can be applied in the special case 
%  where $\alg$ is a worst-case $\approxratio$-approximation algorithm, 
%  though $\ialg$ is only guaranteed to be a $\approxratio$-approximation 
%  in the (weaker) Bayesian sense.
%}
We would
like to use $\alg$ to construct a monotone algorithm $\ialg$ from which we
can obtain a BIC mechanism by simply computing the payment
rule as above.  We must make sure that in doing so we do not reduce the
algorithm's expected welfare.  The key property of our approach which
makes it tractable is that we monotonize each agent's allocation rule
independently without changing (in a Bayesian sense) the allocation rule
any other agent faces.  This property is also important for the
approximation factor as $\alg$ is guaranteed to be a Bayesian
$\approxratio$-approximation only for the given distribution $\dists$, and may not
be a good approximation for some other distribution.

\pagebreak[1]

In summary, the desiderata for monotonizing agent $i$ are:
\begin{enumerate}
\item[D1.] monotone $\ialloci(\vali)$,
\item[D2.] (weakly) improved social welfare $\expect[\vali]{\vali\ialloci(\vali)}
  \geq \expect[\vali]{\vali\alloci(\vali)}$, and
\item[D3.] other agents unaffected.
\end{enumerate}
Notice that if we satisfy the last condition we can apply the process
simultaneously to all agents.

\subsection{Ironing via Resampling}
\label{sec.iron}

There is a history of fixing non-monotonicities in Bayesian mechanism
design.  %Notably, 
Myerson invented the technique of {\em ironing} which
relies on the fact that if an allocation rule is constant over some
interval then any agent within that interval is effectively
equivalent to a canonical ``average'' agent from that interval.
Myerson applied this theory to iron {\em virtual valuation functions}
which are 
%a crucial construct 
used
in Bayesian profit maximization
\cite{mye-81}.  We will apply this theory directly to allocation
rules.

Before we describe our ironing procedure in full, let us develop some more
intuition.  Suppose allocation rule $\alloci(\cdot)$ of $\alg$ is
non-monotone for agent $i$. A simple approach to flattening non-monotonicities is to choose some
interval $[a,b]$ on which $\alloc(\cdot)$ is non-monotone, and to
treat the agent identically whenever on this interval.  For example,
whenever $\vali \in [a,b]$ we could choose to pretend that $\vali$ is
actually some other fixed value $\val'$ (e.g.~$\val' = a$) and pass
this ``pretend'' value $\val'$ to the algorithm.  Unfortunately, if
we take this na\"ive approach, we would have changed the distribution
of agent $i$'s input to the algorithm (in particular, the probability of
value $\val'$ would be increased) and violated D3.  In order to maintain
D3 we make a minor modification: instead of picking a fixed $\val'$, we will
draw $\val'$ from $\disti$ restricted to the interval $[a,b]$.  Thus, we are
replacing $\vali \in [a,b]$ with $\val'$ drawn from the same distribution. 
Other agents cannot tell the difference -- this operation does not change the
distribution of agent $i$'s input!  Moreover, agent $i$ will indeed be treated
identically whenever $\vali \in [a,b]$: the new probability of allocation will
be precisely the distribution weighted average of $\alloci(\cdot)$ over the
interval $[a,b]$.

\begin{figure}
\begin{center}
\begin{tabular}{cc}
\includegraphics[width=2.8in]{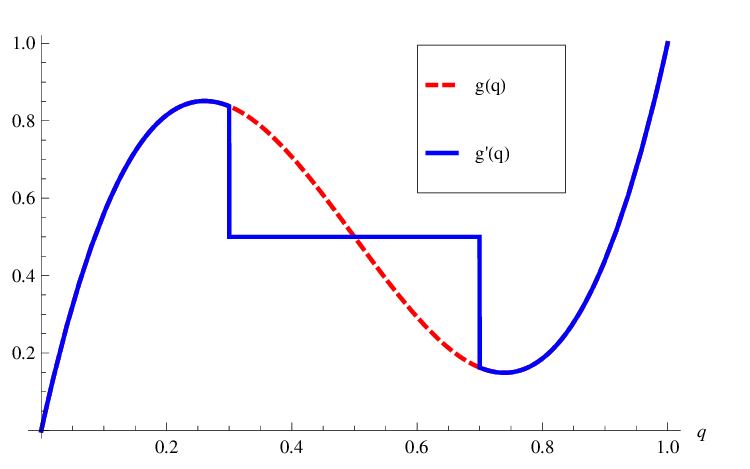} &
\includegraphics[width=2.8in]{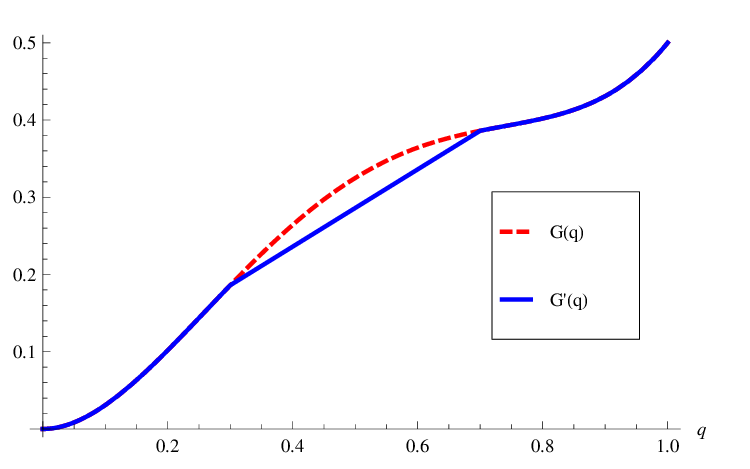} \\
(a) &
(b)\\
\end{tabular}
\end{center}
\caption{\small (a) A non-monotone ironing $g'$ (solid) of curve $g$ (dashed).
  (b) The corresponding integral curves $G'$ (solid) and $G$ (dashed) in probability
  space.}
\label{fig:a-b-ironing}
\end{figure}

%Consider a single agent with value from distribution
%$\dist$ and non-monotone allocation rule $\alloc(\cdot)$.  It will be
%useful to consider this allocation rule in probability space, $g(q) =
%\alloc(\dist^{-1}(q))$, and the {\em cumulative allocation rule} $G(q)
%= \int_0^q g(z)dz$.  Notice that monotonicity of $\alloc(\cdot)$ is
%equivalent to monotonicity of $g(\cdot)$ which is equivalent to
%convexity of $G(\cdot)$.

%\pagebreak

Let $\alloci'(\cdot)$ represent the allocation rule obtained from the
following procedure (where $\valsmi \sim \distsmi$):
\begin{itemize}
\item if $\vali \in [a,b]$, redraw $\val' \sim \disti$ restricted to $[a,b]$; else, set $\val' = \vali$.
\item run $\alg(\val',\valsmi)$.
\end{itemize}

We say that $\alloci'(\cdot)$ is the curve $\alloci(\cdot)$ \emph{ironed on interval $[a,b]$}.  We note that $\alloci'(\vali) = \alloci(\vali)$ for $\vali \not\in [a,b]$ and $\alloci'(\vali) = \expect[\val'\sim\disti]{\alloci(\val')~|~\val'\in[a,b]}$ otherwise.  Note that we can easily iron along multiple disjoint intervals, redrawing $\val'$ from whichever interval contains $\vali'$ (if any).

We now explore a method for choosing intervals on which to iron in order to
obtain monotonicity.  It will be instructive to consider the \emph{allocation rule in probability space} instead of valuation space, and the \emph{cumulative allocation rule} (also in probability space).
\begin{itemize}
\item Let $g(q) = \alloci(\disti^{-1}(q))$ be the allocation rule in probability space.
\item Let $G(q) = \int_0^q g(z)dz$ be the cumulative allocation rule.
\end{itemize}
Notice that monotonicity of $\alloci(\cdot)$ is
equivalent to monotonicity of $g(\cdot)$ which is equivalent to
convexity of $G(\cdot)$.

Let $\alloci'(\cdot)$ be $\alloci(\cdot)$ ironed along some interval $[a,b]$, and consider the corresponding curves $g'(\cdot)$ and $G'(\cdot)$.
This ironing procedure corresponds to replacing $g(\cdot)$
with its average on $[a,b]$, or equivalently $G(\cdot)$ with the line segment connecting $G(\dist(a))$ to $G(\dist(b))$ (See
Figure~\ref{fig:a-b-ironing}).\footnote{Note that the transformation
  to probability space (from valuation space) is necessary for
  obtaining this line-segment interpretation.}  This latter line
segment interpretation suggests that we can view our interval selection problem as the problem of replacing portions of curve $G$ with straight line segments so that the resulting curve $\Gbar$ will be convex.  This is precisely the problem of finding the convex hull of $G$!  Thus the choice of intervals that monotonizes $\alloci(\cdot)$ (satisfying~D1) is precisely the set of intervals defined by the convex hull of $G(\cdot)$.  See
Figure~\ref{fig:mono-ironing}.

Finally, since the convex hull of $G(\cdot)$ lies below $G(\cdot)$, 
%by analogy to first-order stochastic dominance
this transformation weakly improves welfare (satisfying~D2).  Informally
speaking, in moving from cumulative allocation rule $G(\cdot)$ to $\Gbar(\cdot)$, we lower the probability of low-value allocations in
exchange for a corresponding increase in the probability that 
higher-valued allocations occur.  This intuition is made more precise in Lemma \ref{lem.ideal.approx}, below.

\begin{figure}
\begin{center}
\begin{tabular}{cc}
\includegraphics[width=2.8in]{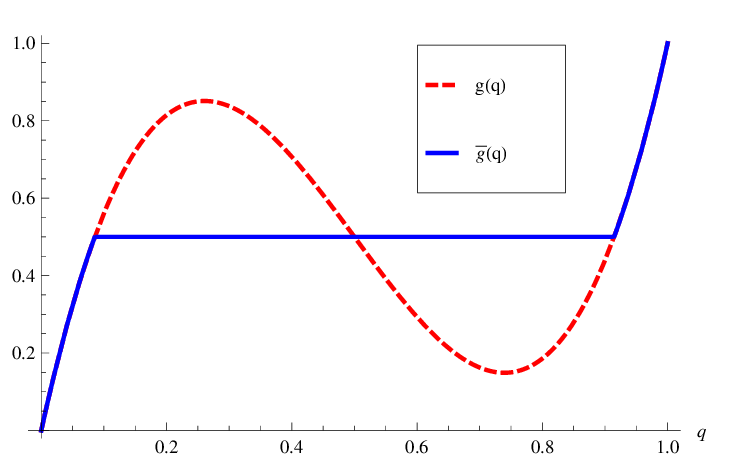} &
\includegraphics[width=2.8in]{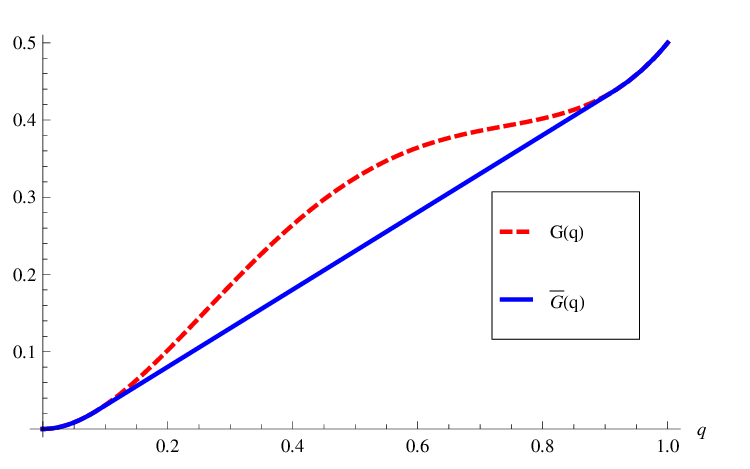} \\
(a) &
(b)\\
\end{tabular}
\end{center}
\caption{\small (a) A monotone ironing $\gbar$ (solid) of curve
  $g$ (dashed).  (b) The corresponding integral curves $\Gbar$
  (solid) and $G$ (dashed) in probability space.  Note $\Gbar$ is
  the convex hull of $G$. }
\label{fig:mono-ironing}
\end{figure}

\subsection{The Ironed Algorithm}
\label{sec.ironed}

We are now ready to define our ironed algorithm $\ialg$.  Given 
distribution $\dist$ and interval $I$, we will write ${\dist}[I]$ 
to mean $\dist$ restricted to $I$.

\begin{definition}[$\resampalg{\alg}$]
\label{def:resampalg}
Given algorithm $\alg$ and a {\em profile of disjoint interval sets}, $\intsets = \{\cali_1,\dotsc,\cali_n\}$, the {\em resampled algorithm} for $\alg$ with intervals $\intsets$ is algorithm $\resampalg{\alg}$:
\begin{enumerate}
\item For each agent $i$, if $\vali \in I \in \intseti$,
  draw $\ivali~\sim {\disti}[I]$; else, set $\ivali = \vali$.
\item Run $\alg(\ivals)$.
\end{enumerate}
\end{definition}
\begin{definition}[$\monoints{\allocs}$]
\label{def:intset}
The {\em set of monoton\-izing intervals for $\allocs(\cdot)$} is
  $\monoints{\allocs} = (\intseti[1],\ldots,\intseti[n])$ defined by:
\begin{enumerate}
\item Let $g_i(q) = \alloci(\disti^{-1}(q))$ be the allocation rule in probability space.
\item Let $G_i(q) = \int_0^q g_i(z)dz$ be the cumulative allocation rule.
\item Let $\Gbar_i(\cdot)$ be the convex hull of $G_i(\cdot)$.
\item Let $\intseti$ be the set of intervals in valuation
  space on which $G_i(\disti(\cdot)) > \Gbar_i(\disti(\cdot))$.
\end{enumerate}
\end{definition}
\begin{definition}[$\ialg$]
\label{def:ialg}
The {\em ironed algorithm} corresponding to algorithm $\alg$ is the algorithm $\ialg = \resampalg[\monoints{\allocs}]{\alg}$.
\end{definition}

\begin{lemma} 
\label{lem.ideal.monotone}
$\ialg$ is monotone.
\end{lemma}

\begin{proof}
We must show that each agent has a monotone allocation rule.
%By construction 
The allocation rule for agent $i$ is precisely $\ialloci(\vali) =
\gbar(\disti(\vali))$, which is the derivative of a convex function and 
therefore monotone.
\end{proof}

\begin{lemma}
\label{lem.ideal.approx}
If $\alg$ is a Bayesian $\approxratio$-approximation then $\ialg$ is a Bayesian $\approxratio$-approximation.
\end{lemma}

\begin{proof}
First notice that
the two allocation rules produce the same distribution over
allocations and therefore expected costs are identical.
We will show, for a single agent $i$, that $\expect{\vali \ialloci(\vali)}
\geq \expect{\vali \alloci(\vali)}$, from which linearity of
expectation implies the result.
We have 
\begin{align*}
\expect{\vali \alloci(\vali)} & = \int_0^1 v \alloci(v) \densi(v)\,dv =  
\int_{0}^{1}\dist_i^{-1}(q)g_i(q)\,dq \\
& = \int_0^1 \int_0^{\dist_i^{-1}(q)}g_i(q)\,dz\,dq = \int_{0}^{1}\int_{\disti(z)}^{1}g_i(q)\,dq\,dz \\
& = \int_{0}^{1}\left(G_i(1)-G_i(\disti(z))\right)\,dz
\end{align*}
and similarly $\expect{\vali \ialloci(\vali)} = \int_{0}^{1}(\Gbar_i(1)-\Gbar_i(\disti(z)))dz$.  We conclude 
$\expect{\vali \ialloci(\vali)} \geq \expect{\vali \alloci(\vali)}$ 
since $\Gbar_i(1) = G_i(1)$ and $\Gbar_i(\disti(z)) \leq G_i(\disti(z))$ for all 
$z \in [0,1]$.
\end{proof}

Theorem~\ref{t:ideal} follows from
Lemmas~\ref{lem.ideal.monotone} and~\ref{lem.ideal.approx}.  

% Notes
\paragraph{Notes.}
We make the following notes about our main result.  A more detailed
discussion is given in our conclusions.
\begin{itemize}
\item The argument fails for non-linear objectives such as makespan.
  While D2 holds, it will not lead to an overall bound on the
  expected performance.  See Appendix~\ref{app:makespan} for an
  example.

\item Our ironing procedure is distinct from Myerson's, in the sense
  that Myerson's procedure yields a different mechanism.  Myerson irons
  virtual valuations and allocates to maximize ironed virtual value.
  This is not the same as maximizing virtual value and then ironing
  the allocation rule where it is non-monotone.  See
  Appendix~\ref{app:Myerson} for an example.

\item Even if $\alg$ is a worst-case $c$-approximation, $\ialg$ may fail to
  be a worst-case $c$-approximation.  See
  Appendix~\ref{app:Bayesian-approx} for an example.
\end{itemize}

\section{Reduction: Black-Box Model}
\label{sec:black-box}

We now turn to a setting in which we do not have full functional access to allocation rules, but only black-box access to the given
algorithm $\alg$ and valuation distribution $\dists$.
We will use the ironing procedure from the previous section to monotonize an algorithm in this black-box model.  
Instead of using direct knowledge of the allocation rule, we must use sampling to estimate it.  This sampling introduces errors in the selection of interval sets for resampling, which must then be dealt with.  
Our analysis will proceed in the following steps.

\begin{enumerate}
\item We describe a method for computing payments in the black-box model.
\item We describe a method for combining sampling with ironing to obtain
  a nearly monotone algorithm.  In fact, this algorithm will be
  $\epsilon$-Bayesian incentive compatible.
\item We show that a convex combination of this nearly monotone
  algorithm with a blatantly monotone one will give a monotone algorithm,
  resulting in a BIC mechanism.
\end{enumerate}

All of these steps approximately preserve social welfare.  We obtain the following theorem.

\begin{theorem}
\label{thm.main.bic.2}
In the black-box model and general cost settings, for any $\epsilon > 0$,
a BIC algorithm 
$\alg'$
can be computed from any 
algorithm $\alg$.  Its expected social welfare satisfies 
$\alg' \geq \alg - \epsilon$, and its
runtime is
polynomial in $n$ and $1/\epsilon$.
\end{theorem}

The additive error in Theorem \ref{thm.main.bic.2} can be converted into a multiplicative error whenever the expected welfare of $\alg$ is not too small.  We obtain the following corollary.

%\begin{theorem}
%\label{thm.main.bic}
%In the black-box model and general cost settings, for any $\epsilon > 0$,
%a BIC algorithm $\corralg{\eps}{\alg}$ can be computed from any Bayesian
%algorithm $\alg$.  Its social welfare satisfies 
%$\corralg{\eps}{\alg}ilon \geq \alg - \epsilon\mumax$, and its
%runtime is
%$\softO(n^{9}\epsilon^{-9}\log^5(\val_{max}/\epsilon\mumax))$.  
%In general feasibility settings, the runtime can be improved to
%$\softO(n^5\epsilon^{-5}\log^3(\val_{max}/\epsilon\mumax))$.
%\end{theorem}

%For the special case of downward-closed set
%systems for feasibility problems, we can assume that $\alg \geq \mumax$, since the trivial
%algorithm that simply allocates to the single player with the highest
%input value attains this value.  This implies the following corollary.

\begin{cor}
\label{cor.main.bic}
In the black-box model and general cost settings,
for any $\epsilon > 0$, 
a BIC algorithm
$\alg'$
can be computed from any algorithm $\alg$.  Its expected 
social welfare satisfies $\alg' \geq \alg/(1+\eps)$, and its
runtime is
polynomial in $n$, $1/\epsilon$, and $1/\alg'$.
\end{cor}

Corollary \ref{cor.main.bic} gives a construction with a multiplicative 
error in social welfare, but its runtime depends on the expected welfare
of $\alg$.  In Appendix~\ref{app:gen-vals} we describe an improvement that
removes this dependency, and implies a fully polynomial reduction 
for downward-closed settings.

%  Specifically, one can show that for any $\eps > 0$ and 
%  algorithm $\alg$, it is possible 
%  to construct a BIC algorithm $\alg'$ such that $\alg' 
%  \geq \alg - \eps \max_i \expect{\vali}$, with runtime polynomial in $n$,
%  $1/\eps$, and $\log(1/\max_i \expect{\vali})$.  For the special case of
%  downward-closed set systems for feasibility problems, this implies that 
%  a BIC algorithm $\alg'$
%  can be computed from any algorithm $\alg$,
%  again $\alg' \geq \alg/(1+\eps)$ and with the runtime of $\alg'$ 
%  being polynomial in $n$, $1/\eps$, and 
%  $\log(1/\max_i \expect{\vali})$.

In the remainder of this section we prove of Theorem \ref{thm.main.bic.2}.

\subsection{Computing Payments}

Suppose that $\alg$ has monotone allocation rules.  The problem of
designing a mechanism to implement $\alg$ then reduces to calculating
appropriate payments.  These payments are completely determined by
the allocation rule of $\alg$, but in the black-box model we do not have
direct access to the functional form of the allocation rule.  Archer
et al.~\cite{APTT-03} solve this problem by computing an unbiased
estimator of the desired payment rule using only black-box calls to the
algorithm.  For completeness we now summarize their approach.

\begin{definition}[black-box payments] 
If algorithm $\alg$ does not allocate to agent $i$, then agent $i$ pays $0$.
Otherwise, we compute the payment
of agent $i$ as follows:
\begin{enumerate}
\item Choose $\vali'$ uniformly from $[0,\vali]$
\item Draw $\valsmi' \sim \distsmi$ and run $\alg(\vali',\valsmi')$
\item If $\alg$ allocated to agent $i$ in the previous step set $X = \vali$, otherwise set $X = 0$.
\item If $X \neq 0$, repeatedly draw values $\valsmi' \sim \distsmi$ and run
  $\alg(\vali,\valsmi')$ until the algorithm allocates to player $i$,
  and let $T$ be the number of iterations required.
\item Agent $i$'s payment is $\pricei = \vali - TX$.
\end{enumerate}
\end{definition}

As was shown by Archer et al., this computation attains the
appropriate expected payment.

\begin{claim}[Archer et al.~\cite{APTT-03}] 
In the black-box payment procedure, the expected payments are\\
$\pricei(\vali) = \vali
\alloci(\vali) - \int_0^{\vali} \alloci(z) dz$.
\end{claim}

We note that since we execute this procedure for agent $i$ only if he
receives an allocation, which occurs with probability $\alloci(\vali)$, the expected number of calls to $\alg$
for each player is at most
\[\alloci(\vali)\left(1+\tfrac{1}{\alloci(\vali)}\right) \leq 2.\]  
Thus, in expectation, all payments can be computed with $2n$
calls to $\alg$.

Any mechanism paired with the above payment scheme  will be {\em individually rational} (IR), meaning
that a truthtelling agent will never obtain negative utility.  This is
true even if the allocation rule is not monotone.  This follows
immediately from the fact that the payment for an agent that declares
value $v_i$ is never greater than $v_i$ (indeed, it is defined as
$v_i$ minus a non-negative value).

\subsection{Sampling and $\epsilon$-Bayesian Incentive Compatibility}
\label{sec:black-box-bne}

We will be estimating the allocation rule of a non-mono\-tone
algorithm and attempting to iron it.  This will fail to result in an
absolutely monotone rule.  In this section we show that a nearly
monotone rule results in truthtelling as an {\em $\epsilon$-Bayes-Nash
  equilibrium }($\epsilon$-BNE): the most an agent can gain from 
a non-truthtelling strategy is an additive $\epsilon$.  We call such a
mechanism {\em $\epsilon$-Bayesian incentive compatible} ($\epsilon$-BIC).
%While there is no characterization of payment rules
%for $\epsilon$-BIC mechanisms, the payment rule given by
%Theorem~\ref{t:bic} will be sufficient.

\begin{definition}[$\epsilon$-BIC]
A mechanism is $\epsilon$-Bayesian incentive compatible if
truthtelling obtains at least as much utility as any other strategy, up to an additive $\epsilon$, assuming all other agents truthtell.  That is,
for all $i$, $\vali$, and $\val'$, 
$\vali \alloci(\vali) -\pricei(\vali) \geq \vali \alloc(\val') - \pricei(\val') - \epsilon.$
\end{definition}

\noindent
The main theorem of this section is the following.

\begin{theorem}
\label{thm.main.eps-bne}
In the black-box model and general cost settings, for any $\epsilon > 0$,
an 
%$\epsilon \val_{max}$-BIC algorithm $\alg'$ 
$\epsilon$-BIC algorithm $\alg'$ 
can be
computed from any algorithm $\alg$. 
Its expected social welfare satisfies
%$\alg' \geq \alg - \epsilon\mumax$. 
$\alg' \geq \alg - \epsilon$,
and its runtime is
%$\softO(n^3\epsilon^{-2}\log_{1+\epsilon}(\val_{max}/\epsilon\mumax))$.
%$\softO(n^3\epsilon^{-3}\log(\epsilon^{-1}))$.
polynomial in $n$ and $1/\epsilon$.
\end{theorem}

The resampling procedure from the previous section is the main
workhorse for Theorem \ref{thm.main.eps-bne}.  The construction of
algorithm $\alg'$ consists primarily of choosing interval sets on
which to resample.

%As in Corollary \ref{cor.main.bic}, for the special case of downward-closed set systems, we can assume
%that $\alg \geq \mumax$ (since any algorithm could simply choose to allocate only to the one player with highest expected value).  This implies the following corollary.

%\begin{cor}
%\label{cor.main.eps-bne}
%In the black-box model, for any $\epsilon > 0$, an $\epsilon
%\val_{max}$-BIC Bayesian $\approxratio(1+\epsilon)$-approximation
%algorithm $\alg'$ can be computed from any Bayesian
%$\approxratio$-approximation algorithm $\alg$ for a downward-closed
%set system.  The running time of $\alg'$ is at most
%$\softO(n^3\epsilon^{-2}\log_{1+\epsilon}(\val_{max}/\epsilon\mumax)$.
%\end{cor}

\subsubsection{$\epsilon$-closeness}

We now formalize a closeness property under which an allocation rule
that is close to monotone 
%(with the payment rule from Theorem~\ref{t:bic}) 
is $\epsilon$-BIC (for some related $\epsilon$).

\begin{definition}[$\epsilon$-close]
Given allocation rules $\alloc(\cdot)$ and $\alloc'(\cdot)$ are
$\epsilon$-close if $\abs{\alloc(\val) - \alloc'(\val) } < \epsilon$
for all $\val$.  Two algorithms or mechanisms are $\epsilon$-close
if each agent's allocation rules are $\epsilon$-close.
\end{definition}

\begin{lemma}
\label{lem.eps-BNE}
If non-monotone $\alg'$ is $\epsilon$-close to a monotone $\alg$,
then $\alg'$ is 
%$(2\epsilon \val_{max})$-BIC.
$(2\epsilon)$-BIC.
\end{lemma}
\begin{proof}
Suppose agent $i$ is participating in $\alg'$ and has value $\vali$, but claims to have value $\vali'$.  Assume $\vali > \vali'$; the opposite case is similar.  Using the payment rule from Theorem~\ref{t:bic}, agent $i$'s gain in utility from declaring $\vali'$ is:
\begin{equation}
\label{eq2}
(\vali \alloci'(\vali') - \pricei(\vali')) - (\vali \alloci'(\vali) - \pricei(\vali)) = (\vali - \vali')\alloci'(\vali') - \int_{\vali'}^{\vali}\alloci'(z) dz.
\end{equation}
Since $\alloci'(\cdot)$ is $\eps$-close to a monotone curve, it must be that $\alloci'(z) + \eps \geq \alloci'(\vali') - \eps$ for all $z \in [\vali',\vali]$. Thus $\int_{\vali'}^{\vali}\alloci'(z) dz \geq (\vali-\vali')(\alloci'(\vali')-2\eps)$. This implies that the value in $\eqref{eq2}$ is at most $2\eps(\vali-\vali')$, which is at most $2\eps$.
\end{proof}

%Now we show that if two algorithms have $\eps$-close allocation rules, then
%they obtain similar welfare.  In the statement of Lemma \ref{lem.welfare} we assume
%that the expected costs of both rules are the same.  This is implied by our techniques: since our mechanism runs the original
%algorithm on values from the original distribution, the expected
%cost of our mechanism is equal to the expected cost of the original
%algorithm.

\begin{lemma}
\label{lem.welfare}
If $\alg$ and $\alg'$ have the same expected costs\footnote{
  Recall that the \emph{expected cost} of an algorithm $A$ with 
  allocation rule $\allocs(\cdot)$ is 
  $\expect[\vals \sim \dists]{\cost(\allocs(\vals))}$.
} and are
$\epsilon$-close then 
%$\alg' \geq \alg - n\epsilon\mumax$.
$\alg' \geq \alg - n\epsilon$.
\end{lemma}
\begin{proof}
For each agent $i$, $\expect{v_i\alloci'(v_i)} \geq \expect{v_i(\alloci(v_i) - \epsilon)} \geq \expect{v_i\alloci(v_i)} - \epsilon\expect{v_i}$.  The result then follows by linearity of expectation.
\end{proof}

%Finally, we note that the property of $\eps$-closeness persists when algorithms are ironed on some collection of intervals.

\begin{lemma}
\label{lem.close.iron}
If algorithms $\alg$ and $\alg'$ are $\epsilon$-close, then for any collection $\intsets = \{\intseti[1],\dotsc,\intseti[n]\}$ of interval sets, resampled algorithms $\resampalg{\alg}$ and $\resampalg{\alg'}$ are $\epsilon$-close.
\end{lemma}
\begin{proof}
For any $i$, let $\alloci$ and $\alloci'$ be the allocation rules of $\alg$ and $\alg'$, respectively.  Let $\ialloci$ and $\ialloci'$ be the allocation rules of $\resampalg{\alg}$ and $\resampalg{\alg'}$.  Then for any $I \in \intseti$, 
\[|\ialloci(I) - \ialloci'(I)| = \left|E_v[\alloc(v)\ |\ v \in I] - E_v[\alloc'(v)\ |\ v \in I]\right| = \left|E_v[\alloc(v) - \alloc'(v)\ |\ v \in I]\right| < \eps.\]
\end{proof}

Appropriate payments to turn an algorithm that is $\epsilon$-close to
monotone into a mechanism that is $2\epsilon$-BIC can be computed by
the same process we would use for monotone algorithms.

\subsubsection{Discretization}

A key step in our reduction will be in discretizing the allocation
rules of the algorithm.  This reduces the problem of estimating an
allocation rule to estimating its value at a polynomial number of
points.  Moreover, our resulting allocation will not necessarily be
monotone, but there will be only a polynomial number of points at
which it can be non-monotone; we will use this to our advantage when
fixing non-monotonicities in Section \ref{sec:black-box-bic}.

%This is needed for two reasons:
%
%\begin{itemize}
%\item We need to estimate the allocation rule and so we need to pick a
%polynomial number of points on it to estimate.  (This is the focus of the next step.)
%\item Our resulting allocation will be almost monotone.  By
%  discretizing we will end up with a polynomial number of points in
%  which it is non-monotone and we will know exactly where these points
%  are.  (In a subsequent step, a convex combination of the sampled
%  algorithm with a blatantly monotone one will fix the
%  non-monotonicities.)
%\end{itemize}

\begin{definition}[Piecewise constant] An algorithm is \emph{$k$-piece piecewise constant} if for each $i$ there is a partition of valuation space into at most $k$ intervals such that the allocation rule for agent $i$ is constant on each interval.  
\end{definition}
%Our discretization procedure, which converts any algorithm into a piecewise constant algorithm, is the following.

\begin{definition}[$\discalg{\alg}$]
\label{def:discalg}
For a given $\epsilon > 0$ and algorithm $\alg$, the \emph{discretization of algorithm $\alg$}, $\discalg{\alg}$, is $\resampalg{\alg}$, where $\intsets = \{\intseti[1],\dotsc,\intseti[n]\}$ is the collection of intervals defined by 
\[ \intseti = \left\{ \left[0,\eps\right) \right\} \cup  \left\{\left[\eps(1+\epsilon)^t,\eps(1+\epsilon)^{t+1}\right)\right\}_{0 \leq t \leq \log_{1+\eps}(1/\eps)}.\]
\end{definition}

\begin{lemma}
\label{lem.useful.discrete}
$\discalg{\alg}$ is
$\log_{1+\epsilon}(1/\epsilon)$-piece 
piecewise constant
and 
$\discalg{\alg} \geq \alg - 2n\epsilon$.
\end{lemma}
\begin{proof}
  Let $\discallocs(\cdot)$ denote the allocation rules for $\discalg{\alg}$.
  The allocation curves for $\discalg{\alg}$ are constant on
  interval $[0,\eps)$ and all intervals of the form
  $[\eps(1+\epsilon)^t,\eps(1+\epsilon)^{t+1})$, and there 
  are at most
  $\log_{1+\epsilon}(\eps^{-1})$ such intervals over the
  range $[\eps,1]$.  These intervals do, indeed, 
  partition valuation space.  Furthermore,
  $\expect[\vali]{\vali \discalloci(\vali)} \geq
  (1-\epsilon)\expect[\vali]{\vali\alloci(\vali)} - \epsilon 
  \geq \expect[\vali]{\vali\alloci(\vali)} - 2\epsilon$, as
  algorithm $\discalg{\alg}$ modifies any input value greater
  than $\epsilon$ by at most a factor of $(1-\epsilon)$.  As
  the expected costs before and after discretization are the same,
  the result follows from linearity of expectation.
\end{proof}

\subsubsection{Statistical Estimation}

We next describe a sampling procedure for estimating an allocation
rule.  
This procedure will not form an algorithm, but
rather generates an estimated allocation curve, which we will denote by $\ests(\cdot)$.
This estimate behaves like an allocation rule, but is not associated
with an actual algorithm (and, in particular, need not be feasibly implementable).  
%Nevertheless, we can pretend that it is an
%algorithm and apply Lemma \ref{lem.eps-BNE} and Lemma \ref{lem.welfare} 
%to it.  The appropriate way to view this is ``if $\ests$ was an allocation rule for an algorithm, then \ldots''

\begin{definition}[estimate allocation rule]
Given algorithm $\alg$ which is $k$-piece piecewise constant and $\epsilon > 0$, an \emph{estimated allocation rule} for $\alg$ is a curve $\ests(\cdot)$ found as follows:
\begin{enumerate}
\item for each agent $i$ and valuation-space piece $I_j$, draw
  $\frac{4}{\epsilon^2}\log{(2kn/\epsilon)}$ samples from $\dists$
  conditional on $\vali \in I_j$, and run $\alg$ on each of these
  samples.
\item let $\est_{ij}$ be the average allocation over the invocations
  to $\alg$ above, for each $i$ and $j$.
\item Define $\ests$ by $\esti(v) = \est_{ij}$ for all $\val \in I_j$
\end{enumerate}
\end{definition}

%Note that $\ests(\cdot)$ is a random variable (even if $\alg$ is deterministic) since it depends on the random bits used during sampling.

\begin{lemma}
\label{lem.useful.sampling}
If algorithm $\alg$ is $k$-piece piecewise constant then, for
any $\epsilon > 0$, an estimated allocation rule $\ests(\cdot)$ for $\alg$ is
$k$-piece piecewise constant, and is $\frac{\epsilon}{2}$-close to
$\allocs(\cdot)$ with probability at least $1 - \frac{\epsilon}{2}$.  The
number of black-box calls to $\alg$ used in the
construction of $\ests(\cdot)$ 
is polynomial in $n$, $k$, and $1/\epsilon$.
%$O(nk(\epsilon)^{-2}\log{(2k/n\epsilon)})$ black-box calls to $\alg$.
\end{lemma}
\begin{proof}
The runtime bound and the fact that $\ests(\cdot)$ is $k$-piece piecewise constant follow immediately from the definition.  Choose some $i$ and let $I_j$ denote piece $j$ of the valuation space for agent $i$ in $\alg$, and write
$\alloci(I_j)$ for the (constant) value of $\alloci(\val)$ for any $\val \in
I_j$.  By the Hoeffding-Chernoff inequality, the probability that $|
\est_{ij} - \alloci(I_j)| > \epsilon/2$ is at most
$e^{-4(\epsilon)^{-2}\log{(2kn/\epsilon)}(\epsilon/2)^2} \leq
\epsilon/2kn$.  Thus, taking the union bound over all $i$ and $j$,
we conclude that
$$\abs{\est_{ij} - \alloci(I_j)} \leq \tfrac{\epsilon}{2}$$ for all $i$ and $j$ with probability at least $1 - \frac{\epsilon}{2}$.
\end{proof}

%\begin{cor}
%With probability $1 - \eps/2$, $\alg_{\ests} \geq \alg - \frac{\eps}{2}n\mumax$ and $\alg_{\ests}$ is $\frac{\eps}{2}$-close to $\alg$.
%\end{cor}

%\subsubsection{Statistical Ironing}

%We are now ready to combine our sampling procedure with the ironing procedure from the ideal model to construct an $\eps$-BIC algorithm from a piecewise constant algorithm $\alg$.

We now complete the proof of Theorem \ref{thm.main.eps-bne} by combining our sampling procedure with the ironing procedure from the ideal model.
\begin{definition}[$\statalg{\alg}$]
Given piecewise constant algorithm $\alg$, the {\em
  statistically ironed algorithm} for $\alg$ with error $\epsilon > 0$ is
$\statalg{\alg} = \resampalg[\monoints{\ests}]{\alg}$ where
$\ests(\cdot)$ is the estimated allocation rule for $\alg$.
%\begin{enumerate}
%\item Construct an estimated allocation rule $\ests(\cdot)$ for $\alg$.
%\item Run $\resampalg[\monoints{\ests}]{\alg}$.
%\end{enumerate}
\end{definition}
Note that $\statalg{\alg}$ is not simply a resampling of $\alg$, but
rather a convex combination of resamplings since the construction of
interval set $\monoints{\ests}$ is randomized.

%Let us give some clarification on step 4 of $\statalg_{\epsilon}$.  Recall that $\ialg$ is the ironed version of some algorithm $\alg$, and $\ialg$ can be constructed in the ideal model.  Since $\statalg_{\epsilon}$ explicitly constructs allocation rule $\ests(\cdot)$, we can think of $\alg_{\ests}$ as an algorithm in the ideal model, so that $\bar{\alg_{\ests}}$ can be determined.  $\bar{\alg_{\ests}}$ defines intervals on which to redraw input values; we will perform this ironing procedure and then pass the modified input to $\discalg_{\epsilon'}$.  Thus, step 4 can be thought of as an ironing of $\discalg_{\epsilon'}$, where the procedure to determine intervals proceeds as though the allocation rule were $\ests(\cdot)$.

\begin{lemma}
\label{lem.statalg}
$\statalg[\epsilon]{\alg}$ is 
%$2\eps\val_{max}$-BIC 
$2\eps$-BIC 
and 
%$\statalg_{\epsilon} \geq \alg - n\eps\mumax$.
$\statalg{\alg} \geq \alg - n\eps$.
\end{lemma}
\begin{proof}
By Lemma \ref{lem.useful.sampling}, $\esti(\cdot)$ 
is $k$-piece
piecewise constant for each $i$.
Let $\alg_{\ests}$ be the (fictional) algorithm with allocation rule $\ests$.
Since $\intsets$ is the monotonizing interval set for $\alg_{\ests}$, if $\alg_{\ests}$ were ironed according to $\intsets$, the result would be $\ialg_{\ests}$ which is monotone.

By Lemma \ref{lem.useful.sampling}, $\alg_{\ests}$ is $\frac{\eps}{2}$-close to $\alg$ with probability $1 - \frac{\eps}{2}$.  In this case, Lemma \ref{lem.close.iron} implies $\resampalg{\alg}$ is $\frac{\eps}{2}$-close to $\ialg_{\ests}$.  For the remaining probability, $\frac{\eps}{2}$, we note that $\resampalg{\alg}$ is trivially $1$-close to $\ialg_{\ests}$.  Thus, taking expectation over all possible outcomes of the sampling, we conclude that $\statalg{\alg}$ is $\eps$-close to monotone, and is therefore $2\eps$-BIC by Lemma \ref{lem.eps-BNE}.

Since, with probability $1-\frac{\eps}{2}$, $\resampalg{\alg}$ is $\frac{\eps}{2}$ close to $\ialg_{\ests}$ and $\alg_{\ests}$ is $\frac{\eps}{2}$ close to $\alg$, Lemma \ref{lem.welfare} and Lemma \ref{lem.useful.discrete} imply that, with probability $1-\frac{\eps}{2}$,
\begin{equation*}
%\resampalg{\alg} \geq \ialg_{\ests} - \frac{1}{2}n\epsilon\mumax \geq \alg_{\ests} - \frac{1}{2}n\epsilon\mumax \geq \alg - n\epsilon\mumax.
\resampalg{\alg} \geq \ialg_{\ests} - \tfrac{1}{2}n\epsilon \geq \alg_{\ests} - \tfrac{1}{2}n\epsilon \geq \alg - n\epsilon.
\end{equation*}
For the remaining probability, $\frac{\eps}{2}$, we note that trivially 
$\resampalg{\alg} \geq 0 = \alg - \alg \geq \alg - n$.  
%$\resampalg{\alg} \geq 0 = \alg - \alg \geq \alg - n\mumax$.  
Thus, taking expectation over all possible outcomes of sampling, we conclude 
$\statalg{\alg} \geq \alg - n\epsilon$.
%$\statalg{\alg} \geq \alg - n\epsilon\mumax$.
\end{proof}

%\noindent
%We are now ready to complete the proof of Theorem \ref{thm.main.eps-bne}.
\noindent
\textbf{Proof of Theorem \ref{thm.main.eps-bne}: }
%Theorem \ref{thm.main.eps-bne} now follows by taking 
Define $\alg'$ to be the algorithm $\statalg[\epsilon']{\discalg[\epsilon']{\alg}}$, where $\eps' = \eps/3n$.  Then, by Lemmas \ref{lem.useful.discrete} and \ref{lem.statalg}, $\alg'$ is
%$2\eps'\val_{max}$-BIC, 
$2\eps'$-BIC, 
and hence 
%$\eps\val_{max}$-BIC, 
$\eps$-BIC, 
and 
%$\alg' \geq \discalg_{\eps'} - n\eps'\mumax \geq \alg - 3n\eps'\mumax = \alg - \eps\mumax$. 
$\alg' \geq \discalg[\epsilon']{\alg} - n\eps' \geq \alg - 3n\eps' = \alg - \eps$. 
The runtime of $\alg'$ (which is dominated by sampling in the construction of $\ests$) is 
%$O(nk{\eps'}^{-2}\log{(2kn/\eps'))}) = \softO(n^3\epsilon^{-3}\log(v_{max}/\epsilon\mumax))$, 
$O(nk{\eps'}^{-2}\log{(2kn/\eps'))}) = \softO(n^3\epsilon^{-3}\log(\epsilon^{-1}))$, 
where recall 
%$k = \log_{1+\eps}(v_{max}/\epsilon\mumax)$ 
$k = \frac{1}{\eps}\log(1/\epsilon)$ 
is the number of discrete intervals in $\discalg[\epsilon']{\alg}$.
\QED

\subsection{Bayesian Incentive Compatibility}
\label{sec:black-box-bic}

In the previous section we showed how to construct an $\epsilon$-BIC mechanism from any algorithm with almost no loss to the social welfare.
Our goal now is to take such an $\epsilon$-BIC algorithm $\alg$ and make it BIC.  In other words, we would like to ``fix'' the (small) non-monotonicities in $\alg$.  Fortunately, since each allocation curve of $\alg$ is discretized,  any non-monotonicities must occur only at a small number of predetermined points.  Our approach for removing these points of non-monotonicity is simple: we will construct an alternative algorithm $\alg'$ whose allocation curves are stair functions, with jumps in allocation probability occurring at each of those points.  A convex combination of $\alg$ and $\alg'$ will then be monotone.  This convex combination will be our final BIC algorithm.

It is important that this convex combination process not reduce social
welfare by too much.  This requires two things.  First, we need the
convex combination to be mostly $\alg$ as only it has provably good
welfare.  This is possible by taking $\epsilon$ so small that the
explicit monotonicities in $\alg'$ heavily outweigh the
non-monotonicities in $\alg$ (which are at most $\epsilon$).  Second,
we need to ensure that the expected social welfare of $\alg'$ is not
extremely negative.

How should we construct $\alg'$?  Suppose first that we are 
in a downward-closed feasibility setting.  In this case, the singleton allocation $\{i\}$ is feasible for each agent $i$.  The construction of $\alg'$ with stair-function allocation curves is then straightforward: an agent $i$ is chosen uniformly at random and the algorithm then either allocates to agent $i$ or not, with the probability of allocation following a stair function.  Since $\alg'$ only returns feasible outcomes, its expected social welfare must be non-negative.

We would like to follow this same approach in general cost settings.  However, it may be that, for some $i$, the particular allocation $\{i\}$ has an extremely high (or infinite) cost, in which case the above algorithm may have an extremely negative social welfare.  Note, though, that in our construction we can replace $\{i\}$ with \emph{any} allocation that includes agent $i$.  It is therefore sufficient to find, for each $i$, some allocation that includes agent $i$ and whose cost is not too high.  Once these allocations are found, we can use them to construct the stair algorithm $\alg'$.

In some cases finding low-cost allocations may be highly non-trivial.
To get around this problem, we observe that as long as algorithm
$\alg$ has a reasonable probability of allocating to agent $i$, there
must exist low-cost allocations that include $i$ that are returned by
$\alg$.  We can therefore find such allocations by repeatedly sampling
outcomes of $\alg$.  If, on the other hand, we were to take many
samples and not find any allocations that include agent $i$, then we
can safely assume that agent $i$ does not contribute much to the
expected social welfare of $\alg$.  In this case, we can trivially
monotonize agent $i$'s allocation curve by ironing on interval $[0,1]$,
removing the need to find allocations that include him.  

\subsubsection{The Stair Algorithm}

%In the previous section we showed how to construct a $\epsilon$-BIC mechanism from any algorithm with almost no welfare loss.  Furthermore,
%$\epsilon$ can be made arbitrarily small with more sampling.  
%We now outline machinery by which we can make an $\epsilon$-BIC 
%mechanism BIC, outright, without much additional loss.  
%This 
%approach will apply to any $\epsilon$-BIC mechanism with allocation 
%rules that are piecewise constant.  
%
%
%We assume that we can easily find a feasible
%allocation for any agent who is allocated with non-negligible
%probability by the original algorithm. (Notice that any sampling
%procedure that adequately samples the the allocation rules must
%uncover these allocations.  Any agents not satisfying this property can
%be ignored outright.)  
%
%
%We note that this approach is not specific to
%our mechanism, and can generally be applied to any piecewise constant
%mechanism that is close to another monotone mechanism.
%
%\subsubsection{An Explicitly Monotone Algorithm}
%
%Our approach will be to form a convex combination of a piecewise constant algorithm $\alg$ that is close to monotone with the following explicitly monotone algorithm.

We begin by demonstrating how to combine an $\eps$-BIC mechanism with an algorithm whose allocation rules are stair functions in order to obtain a BIC mechanism.

\begin{definition}[$\stair{\alg}$] Let $\alg$ be a $k$-piece piecewise constant algorithm, and suppose $S_1, \dotsc, S_n$ and $T_1, \dotsc, T_n$ are allocations such that $i \in S_i$ and $i \not\in T_i$ for all $i$.  The {\em stair algorithm for $\alg$}, $\stair{\alg}$, does the following:
\begin{enumerate}
\item Pick an agent $i$ uniformly from the $n$ agents.
\item If $\vali$ is in the $j$th highest piece of $k$ pieces, allocate
  to $S_i$ with probability $(j-1)/(k-1)$ and $T_i$ otherwise.
\end{enumerate}
\end{definition}
%
%We have not yet specified how each allocation $S_i$ is chosen; the process will be different for feasibility settings and general cost settings, as we describe below.  First, we note that if these allocations $S_1, \dotsc, S_n$ can be found, then a convex combination of $\alg$ and $\stair{\alg}$ will be monotone.
%
%If $\alg$ is $\eps$-close to a monotone algorithm, a convex combination of $\alg$ and $\stair{\alg}$ will be monotone.
%
%depend on the problem setting.  For feasibility settings we will require that $S_i$ be feasible, whereas in general cost settings we will need to bound the cost of $S_i$.  We will handle these settings separately below, but first we present the main technical Lemma of this section.

%We note that the stair algorithm makes allocation $S_i$ with positive probability only if $\vali$ is not in the lowest piece of agent $i$'s allocation rule.  This motivates the following definition.

%\begin{definition}[$w_i$] Given algorithm $\alg$ with piece-wise constant allocation rules, the \emph{stair threshold for agent $i$}, $w_i$, is the supremum of all values in the first piece of agent $i$'s allocation rule.
%\end{definition}

\begin{definition}[$\combalg{\alg}$] Suppose algorithm $\alg$ is $k$-piece piecewise constant.  Then $\combalg{\alg}$ is the convex combination of 
$\alg$ with probability $1-\stairfrac$ and
$\stair{\alg}$ with probability $\stairfrac$, 
%$\alg$ and $\stair{\alg}$ with probabilities $1-\stairfrac$ and 
%$\stairfrac$, respectively, 
where $\stairfrac =  2(k-1)n\epsilon$.
\end{definition}

\begin{lemma}
\label{lem.step}
If $\alg$ is $\epsilon$-close to a monotone $\alg'$, 
then algorithm $\combalg{\alg}$ is BIC.
\end{lemma}

\begin{proof}
We will write $\corralloci(\cdot)$ to denote an allocation rule of $\combalg{\alg}$.
To show $\combalg{\alg}$ is BIC, choose any agent $i$ and any values $\vali < \vali'$; we will show $\corralloci(\vali) \leq \corralloci(\vali')$.  If $\vali, \vali'$ are in the same piece of the valuation space then $\corralloci(\vali) = \corralloci(\vali')$.  Otherwise, since $\alg$ is $\epsilon$-close to monotone $\alg'$, it must be that $\alloci(\vali) \leq \alloci(\vali')-2\epsilon$.  Furthermore, if $\mathbf{s}(\cdot)$ is the allocation rule for 
$\stair{\alg'}$, then $s_i(\vali) \leq s_i(\vali')+1/(k-1)n$.  We conclude that 
\begin{equation*}
\begin{split}
\corralloci(\vali) & = (1-\stairfrac)\alloci'(\vali) + \stairfrac s_i(\vali) \\
& \leq \corralloci(\vali') - 2\epsilon + \stairfrac/(k-1)n \\
& = \corralloci(\vali')
\end{split}
\end{equation*}
as required, since $\stairfrac = 2(k-1)n\epsilon$.
\end{proof}

\subsubsection{Bounding Social Welfare: Finding Low-Cost Sets}

%We would like to apply Lemma \ref{lem.step} to $\statalg{\alg}$ from
%Lemma \ref{lem.statalg}, thereby proving Theorem \ref{thm.main.bic.2}.
We now describe the choice of sets $S_1,\dotsc, S_n$ and $T_1,\dotsc,T_n$ for algorithm $\stair{\alg}$.
What we require is that, for all $i$, $i \in S_i$, $i \not\in T_i$, and $S_i$, $T_i$ are feasible (or have sufficiently low cost).
In many settings finding such sets is trivial (e.g., for downward-closed feasibility problems we can take $S_i = \{i\}$ and $T_i = \emptyset$), 
%, we could simply take $S_i = \{i\}$ and Theorem
%\ref{thm.main.bic} follows immediately.  %However, in general settings,
%it may be difficult to find feasible sets (or low-cost sets in general
%cost settings).  The remainder of this section will be devoted to
%describing a general process by which sets $S_1, \dotsc, S_n$ can be
%found.
%
%\subsubsection{Implementing the Stair Algorithm in General Cost Settings}
%
but for some problems it might be difficult to find feasible (or low-cost) allocations.
% or algorithm $\stair{\alg}$ may incur negative
%value if, for some $i$, the cost of set $S_i$ is large relative to
%$\vali$.  Thus, to preserve the approximation ratio of our construction, we must find sets $S_i$ with low costs.
%
Our approach is as follows.  Since $\alg$ never makes an allocation that generates negative social welfare, we can bound the cost of any allocation made by $\alg$.  This motivates us to look for a set $S_i \ni i$ returned by $\alg$ on some input, for each $i$.  This can be accomplished by sampling. %techniques used in the construction of $\statalg{\alg}$.
%That is, for each $i$ and each piece of the valuation space, we will
%run $\alg$ on many sample inputs.  As long as $\alloci(I)$ is not too
%small on a given interval $I$, we are very likely to find some valid
%allocation that includes agent $i$ during the sampling process!  
In the event that we do not find a set $S_i$, it is likely that the probability of allocating to agent $i$ is very low; we can therefore \emph{iron together all intervals} for agent $i$, effectively removing the need for $S_i$, without causing much loss to the expected welfare.
This operation can be viewed as trimming away agents that are very rarely allocated.  The same holds for finding $T_i$.

\begin{definition}[$\trimalg{\alg}$]
\label{def:stairalg}
%Given piece-wise constant algorithm $\alg$, 
The \emph{trimmed algorithm} for piece-wise constant $\alg$ is $\trimalg{\alg}$:
\begin{enumerate}
\item For each agent $i$ and valuation-space piece $I_j \in \cali_i$, draw $\frac{4}{\eps^2}\log(2n/\eps)$ samples from $\dists$ conditional on $\vali \in I_j$, and run $\alg$ on each of these samples.
\item If $\alg$ is the same (always or never allocating) for $i$ on every sample, define $\intseti' = \{ [0,1] \}$; otherwise, $\intseti' = \intseti$ and we define $S_i$ to be any observed allocation that includes agent $i$ and $T_i$ to be any observed allocation that does not include agent $i$.
%\item \quad Let $j_i$ be the minimal index such that, for some sample of interval $I_{j_i}$, $\alg$ allocated a set $T_i$ with $T_i \ni i$ and $\cost(T_i) \leq \min I_{j_i} + n\mumax/\sqrt{\eps}$.  Choose $S_i$ to be any such $T_i$.
%\item \quad If no such set $T_i$ was returned for \emph{any} interval, take $j_i = k+1$ and $S_i = \{i\}$.
%\item \quad Define $\intseti' = \{ I_1 \cup \dotsc \cup I_{j_i-1}, I_{j_i}, \dotsc, I_k \}$.
\item Run $\resampalg[\intsets']{\alg}$.
\end{enumerate}
\end{definition}

Note that, for each $i$, either sets $S_i \ni i$ and $T_i \not\ni i$ will be found during the execution of $\trimalg{\alg}$, or else the allocation rule of agent $i$ will be made constant. %, which is taken to be any set satisfying the conditions on line 4 for interval $I_{j_i}$ (or $\{i\}$ if no sets were found).

%In summary, $\trimalg_\eps$ samples each constant interval for agent $i$, searching for an appropriate set $S_i$.  We take $I_{j_i}$ to be the leftmost interval for which such a set $S_i$ was found.  All intervals to the left of $I_{j_i}$ are then ironed together.  Thus, regardless of the sampling outcome, $I_{j_i}$ will be the second valuation space piece for agent $i$ in algorithm $\ialg_{\intsets'}$.  Thus $\cost(S_i) \leq w_i^{\trimalg_\eps} + n\mumax/\sqrt{\eps}$.

%\begin{lemma}
%\label{lem.findsets-bound}
%The stair compatible algorithm $\trimalg_\eps$ for
%$\alg$
%(Definition~\ref{def:stairalg}) and stair thresholds
%${\mathbf w}^{\trimalg_\eps}$ (Definition~\ref{def:stairthreshold}) satisfy $\cost(S_i) \leq w_i^{\trimalg_\eps} +
%n\mumax/\sqrt{\eps}$ for all $i$.
%\end{lemma}
%\begin{proof}
%For each $i$, if no
%set satisfying the conditions on line 4 of $\trimalg_\eps$ was found
%during the sampling of any interval, then $j_i = k+1$ and all
%intervals of $\alg$ are ironed together in $\intsets'$.  In this case
%$w_i^{\trimalg_\eps} = \infty$, so $\cost(S_i) \leq
%w_i^{\trimalg_\eps}$ trivially.  Otherwise, by line 4 of
%$\trimalg_\eps$, $\cost(S_i) \leq \min I_{j_i} +
%n\mumax/\sqrt{\eps}$.  However, since all intervals to the left of
%$I_{j_i}$ are ironed together in $\intsets'$, $I_{j_i}$ will be the
%second piecewise constant interval of $\trimalg_\eps$, and hence
%$\min I_{j_i} = w_i^{\trimalg_\eps}$.  Thus $\cost(S_i) \leq \min
%w_i^{\trimalg_\eps} + n\mumax/\sqrt{\eps}$ as required.
%\end{proof}

\begin{lemma}
\label{lem.findsets}
%The stair compatible algorithm $\trimalg{\eps}{\alg}$ for
%$\alg$
%(Definition~\ref{def:stairalg}) satisfies
%$\trimalg_\eps \geq \alg - 2n\mumax/\sqrt{\eps}$.
$\trimalg{\alg} \geq \alg - n\eps$.
\end{lemma}
\begin{proof}
We claim that, with probability at least $1-\frac{\eps}{2}$, for each
agent $i$, the allocation rules for $\trimalg{\alg}$ and $\alg$ will 
differ only on values $\vali$ for
which $\alloci(\vali) \leq \frac{\eps}{2}$.  Before
proving the claim, let us see how it implies the desired result.  The
claim implies that $\trimalg{\alg} \geq \alg -
(\frac{\eps}{2})n$ with probability $1 -
\frac{\eps}{2}$.  For the remaining probability, we note that
$\trimalg{\alg} \geq 0 = \alg - \alg \geq \alg - n$ trivially.
Thus, over all possible outcomes of sampling, we conclude that
\[ \trimalg{\alg} \geq \alg - \tfrac{\eps}{2}n - \tfrac{\eps}{2}n = \alg - n\eps \]
as required.

Let us now prove the claim.  Choose some agent $i$ and suppose that
$\trimalg{\alg}$ and $\alg$ differ on some interval $I$ with
$\alloci(I) \geq \frac{\eps}{2}$.  
%Let $I$ be the leftmost
%such interval. For the remainder of the proof we will say that a set
%$T$ has \emph{low cost for $I$} if $\cost(T) \leq \min I +
%n/\sqrt{\eps}$.  
Then, by the definition of $\intseti'$, it must be
that no set $T \ni i$ was found during the sampling of
interval $I$ for agent $i$.  However, since $\alloci(I) \geq \frac{\eps}{2}$, there is a probability of at least $\frac{\eps}{2}$ of finding such a set $T$ on each sample.  
%Consider some agent $i$ and some interval $I \in \cali$.  For
%the remainder of the proof we will say that a set $T$ has \emph{low
%  cost for $I$} if $\cost(T) \leq \min I + n\mumax/\sqrt{\eps}$.
%Given $\vals \sim \dists$, let $B(\vals)$ be the event
%$[\alloci(\vals) \wedge \sum_i \vali \leq \min I +
%  n/\sqrt{\eps}]$.  If event $B(\vals)$ occurs for some sample $\vals$,
%this means that $\alg$ returned some allocation $T \ni i$ and
%furthermore $\sum_i \vali \leq \min I + n/\sqrt{\eps}$.  But
%note that this allocation must generate non-negative profit (otherwise
%it would never be allocated), and hence $T$ must have low cost for
%$I$.  Thus $B(\vals)$ is precisely the event that $\alg$ returns a set
%$T \ni i$ with low cost for $I$.
%
%Consider the probability of $B(\vals)$.  By Markov's inequality,
%$\prob[\vals]{\sum_{j \neq i} \val_j > n/\sqrt{\eps}} <
%\sqrt{\eps}$.  Thus, since $\vali \geq \min I$ with probability $1$
%conditional on $\vali \in I$, $\prob[\vals]{\sum_i \vali > \min I +
%  n\mumax/\sqrt{\eps}} < \sqrt{\eps}$.  Also,
%$\prob[\vals]{\neg\alloci(\vals)\ |\ \vali \in I} = 1-\alloci(I) \leq
%1 - (\sqrt{\eps} + \frac{\eps}{2})$.  The union bound then implies
%that $\prob[\vals]{\neg B(\vals)} \leq
%1-(\sqrt{\eps}+\frac{\eps}{2})+\sqrt{\eps} = 1 - \frac{\eps}{2}$, so
%$\prob[\vals]{B(\vals)} \geq \frac{\eps}{2}$.
%
By Chernoff-Hoeffding inequality, the probability that we do not find even one such set during $4\eps^{-2}\log(2n/\eps)$ samples is at most
$\frac{\eps}{2n}$.  We conclude that the probability that no set $T
\ni i$ was found during the sampling of interval $I$ is
at most $\frac{\eps}{2n}$.  This is therefore a bound on the
probability that $\trimalg{\alg}$ and $\alg$ differ for agent $i$
on some interval $I$ with $\alloci(I) \geq
\frac{\eps}{2}$.  By the union bound, the probability that
this occurs for \emph{any} agent is at most $\frac{\eps}{2}$, as
required.
%
%We now wish to argue that, with high probability, the allocation rules for $\alg$ and $\trimalg_\eps$ differ for agent $i$ only on intervals $I$ with $\alloci(I) < \sqrt{\eps}+\frac{\eps}{2}$.  For each $i$, if $\alloci(I) < \sqrt{\eps}+\frac{\eps}{2}$ for all intervals $I$, then this is trivially true.  Otherwise, let $I_{j_i}$ denote the leftmost interval for which $\alloci(I) \geq \sqrt{\eps}+\frac{\eps}{2}$.  By the union bound, with probability $1 - \eps/2$, for each $i$, a set $S_i \ni i$ with low cost will be found when sampling interval $I_{j_i}$.  In this case, the behaviour of algorithms $\trimalg_\eps$ and $\alg$ differ only on intervals to the left of $I_{j_i}$, all of which satisfy $\alloci(I) < \sqrt{\eps}+\frac{\eps}{2}$.  Thus, conditioning on an event of probability $1 - \frac{\eps}{2}$, \[\trimalg_\eps \geq \left(1-\left(\sqrt{\eps}+\frac{\eps}{2}\right)\right)\alg \geq \alg - \left(\sqrt{\eps}+\frac{\eps}{2}\right)n\mumax.\]  
%For the remaining probability, $\frac{\eps}{2}$, we note that $\alg \leq n\mumax$ trivially.  We conclude that 
%\[\trimalg_\eps \geq \alg - \frac{\eps}{2}n\mumax - \left(\sqrt{\eps}+\frac{\eps}{2}\right)n\mumax \geq \alg - 2\sqrt{\eps} n\mumax\] 
%as required.  
\end{proof}

We are now ready to combine our tools into a BIC mechanism, proving Theorem \ref{thm.main.bic.2}.
%$\corralg{\epsilon}{\alg}$ from the statement of Theorem \ref{thm.main.bic.2}.

\begin{definition}[$\corralg{\alg}$] Given an algorithm $\alg$ and $\eps > 0$, the \emph{monotonization of $\alg$}, denoted $\corralg{\alg}$, is the algorithm $\combalg{\statalg{\trimalg{\discalg{\alg}}}}$.
%\begin{enumerate}
%\item Construct $\discalg{\epsilon}{\alg}$, the discretized version of $\alg$.
%\item Construct $\trimalg_{\eps}$, the stair-compatible version of $\discalg{\epsilon}{\alg}$.  This generates sets $S_1, \dotsc, S_n$.
%\item Construct $\statalg{\epsilon}{\alg}$, the statistically ironed algorithm for $\trimalg_{\eps}$.
%\item With probability $\stairfrac = 2k(n-1)\eps$, execute $\stair{\trimalg_{\eps}}$ with sets $S_1, \dotsc, S_n$.  Else, execute $\statalg{\epsilon}{\alg}$.
%\end{enumerate}
\end{definition}

\begin{lemma}
\label{lem.corralg}
$\corralg{\alg}$ is BIC, and 
%$\corralg{\eps}{\alg} \geq \alg - 7kn^2\sqrt{\eps}\mumax$.
$\corralg{\alg} \geq \alg - 6kn^2\eps$.
\end{lemma}
\begin{proof}
For notational convenience we define $\alg' = \trimalg{\discalg{\alg}}$.
Recall that during the construction of $\alg'$ we find sets $S_1, \dotsc, S_n$ with $S_i \ni i$.
Also, Lemma \ref{lem.statalg} implies that $\statalg{\alg'}$ is $\eps$-close to a monotone algorithm.  Thus $\combalg{\statalg{\alg'}}$ is well-defined, and is also BIC by Lemma \ref{lem.step}.  

Our ironing techniques do not affect the distribution of allocations generated by an algorithm, so the expected costs of $\corralg{\alg}$ and $\alg$ are the same.  Furthermore, by Lemmas \ref{lem.useful.discrete}, \ref{lem.statalg}, and \ref{lem.findsets}, 
\[\statalg{\alg'} \geq \alg' - n\eps = \trimalg{\discalg{\alg}} - n\eps \geq \discalg{\alg} - 2n\eps \geq \alg - 4n\eps.\]  
We next claim that one can assume without loss of generality that $\cost(S_i) \leq n$ for all $i$.  This is because $S_i$ is in the range of $\alg$, and we can assume that $\alg$ never returns an allocation that results in negative welfare (since otherwise a trivial improvement to $\alg$ would return the empty allocation instead).  Since valuations lie in $[0,1]$, non-negative welfare can be generated only by sets with cost at most $n$, and thus we can assume $\cost(S_i) \leq n$ for all $i$.
%by Lemma \ref{lem.findsets-bound}, 

This implies that the expected social welfare obtained by $\stair{\statalg{\alg'}}$ is at least $(-n)$.  We conclude
\begin{align*}
\corralg{\alg} & = \combalg{\statalg{\alg'}} \\
& = (1-\stairfrac)\statalg{\alg'} - \stairfrac \stair{\statalg{\alg'}} \\
& \geq \alg - 4n\eps - (2(k-1)n\eps)n \\
& \geq \alg - 6kn^2\eps.
\end{align*}
%Taking $\eps' = (\eps/7kn^2)^2$, we arrive at the desired result.
\end{proof}

%\noindent
%We are now ready to complete the proof of Theorem \ref{thm.main.bic.2}.

\noindent
\textbf{Proof of Theorem \ref{thm.main.bic.2}: }
%Theorem \ref{thm.main.bic.2} now follows immediately from Lemma \ref{lem.corralg} by considering algorithm 
Let $\alg'$ be the monotonized algorithm $\corralg[\epsilon']{\alg}$, where 
%$\eps' = (\eps/7kn^2)^2 = \eps^2/49k^2n^4$.  
$\eps' = \eps/6kn^2$.  The result then follows immediately from Lemma \ref{lem.corralg}.
The runtime, which is dominated by sampling, is $O(kn(\eps')^{-2}) = %\softO(n^9\eps^{-9}\log^5(\val_{max}/\eps\mumax))$.
\softO(\frac{n^5}{\eps^5}\log^3(1/\eps))$.
\QED

\section{Conclusions}

\label{sec:conclusions}

Our main result is for single-parameter agents and the
objective of social welfare where we give a black-box reduction that
converts any Bayesian approximation algorithm into a Bayesian
incentive compatible mechanism.  For these settings there is no gap
separating the approximation complexity of algorithms and BIC
mechanisms.  

It is notable that our transformation from an approximation algorithm
to a BIC mechanism cannot be duplicated by the agents acting on their
own: there are non-monotone algorithms that, when coupled with any
reasonable payment rule, do not have any BNE with near the expected
welfare as the original algorithm on the true values.  A concrete
example is given in Appendix \ref{app:original}.

While our main theorem is extremely general, the situations
not covered by it are of notable interest.

\begin{enumerate}
\item Multi-parameter Bayesian mechanism design is not very well
  understood, but there is every reason to believe that approximation
  (which has not been pursued much by the economics literature) has a
  very interesting and relevant role to play in providing positive
  results.  For the objetive of profit maximization the result
  of~\cite{CHMS-10} reduces mechanism design to algorithm design in
  {\em unit-demand} settings with a natural ``substitutability''
  property of the feasibility constraint, e.g., from matroid set
  systems.  For social welfare maximization, the approach of this
  paper was recently generalized to convert any algorithm to a BIC
  mechanism in multi-dimensional discrete settings~\cite{HKM-11}.  The
  same approach gives an $\epsilon$-BIC approximation in general
  multi-dimensional settings~\cite{BH-11}.  These reductions can be
  applied to the combinatorial public project problem for which
  Papadimitriou et al.~\cite{PSS-08} exhibit a gap separating the
  approximation complexity of algorithms and ex post IC mechanisms.
  These reductions show that the gap is in fact between BIC and IC
  mechanisms~\cite{BH-11}.

\item Our reduction applies to the objective of social welfare
  maximization.  It would be nice to extend our result more generally
  to any monotone objective function, e.g., makespan.  Unfortunately,
  our approach fails to preserve the approximation factor of the
  makespan objective.  A concrete example that shows tha tour ideal
  reduction does not preserve expected makespan is given in
  Appendix~\ref{app:makespan}.  {\em Is there a polynomial-time
    reduction that turns any approximation algorithm for any monotone
    objective into a BIC mechanism with the same approximation
    factor?}

\item In the special case where our reduction is applied to a
  worst-case $\approxratio$-approximation (recall: our reduction applies more
  generally to Bayesian $\approxratio$-approximations), the resulting BIC
  mechanism is still only a $\approxratio$-approx\-i\-mation in the weaker Bayesian sense.
%  A concrete example of this is given in the full version.
%  See 
%  Appendix~\ref{app:Bayesian-approx}.  
{\em
    Is there a poly\-no\-mial-time black-box reduction that turns any worst-case
    $\approxratio$-approximation algorithm into a BIC mechanism that is also a
    worst-case $\approxratio$-approximation?}

\item While Bayes-Nash equilibrium (i.e., BIC) is the standard
  equilibrium concept for implementation in economics, the stronger
  dominant strategy equilibrium (i.e., ex post IC) is the standard
  concept
  in computer science.  The main
  challenge in obtaining a similar reduction for IC mechanisms is that
  the valuation space is exponentially big and monotonizing all points
  seems to require an exhaustive procedure.  
  %This 
  %seems more difficult in 
  %does not seem to be the case in
  %non-Bayesian settings.  
  One potential approach would be to apply our
  ironing technique repeatedly, re-ironing an agent's curve whenever it
  is affected by an ironing of another agent's curve.  
  Such a procedure
  %can be argued to terminate, but 
  will not generally 
  give a monotone allocation rule; %see
  A concrete example is given in the full version of the paper.
  %Appendix~\ref{app:IC-dependence}. 
  {\em Is
    there a polynomial-time reduction for turning any
    $f(n)$-approximation (worst-case or Bayesian) algorithm for a 
    single-parameter domain into an ex
    post IC mechanism that is a (worst-case or Bayesian)
    $\Theta(f(n))$-approximation?}
\end{enumerate}
The final question above is a refinement of what we consider to be the
main open question of this work.  {\em Is there a gap separating the
  approximation complexity of implementation by Bayesian incentive
  compatible and ex post incentive compatible mechanisms for
  single-parameter social welfare maximization?}

\bibliographystyle{plain}
\bibliography{auctions}

\appendix

\section{Ironing Allocation Rules vs. Ironing Virtual Valuations}
\label{app:Myerson}

At the heart of our mechanism construction is an ironing procedure that monotonizes allocation rules, outlined in Section \ref{sec:ideal}.
A similar process is used by Myerson as part of his construction of (revenue) optimal mechanisms 
for single-parameter settings \cite{mye-81}.  
In light of this similarity, we will now compare these two constructions and highlight their differences.

We first recall Myerson's optimal mechanism.  For each agent $i$, the mechanism considers the \emph{virtual valuation function} $\phi_i(\cdot)$
given by $\phi_i(\vali) = \vali - \frac{1 - \disti(\vali)}{\densi(\vali)}$.  This function is monotonized\footnote{Note that the virtual valuation function may be non-monotone if $\disti$ does not satisfy the monotone hazard rate assumption.  For instance, bimodal distributions generally have non-monotone virtual valuation functions.} using the ironing method described in
Section \ref{sec:ideal}; the resulting monotone function is denoted $\overline{\phi_i}(\cdot)$.  Given a valuation profile $\vals$, the mechanism
returns the allocation $\allocs$ that maximizes $\sum_i \overline{\phi_i}(\vali) \cdot \alloci - \cost(\allocs)$.  Myerson's celebrated result is that 
this allocation rule is revenue-optimal among the class of incentive compatible allocation rules.

Informally speaking, one can interpret Myerson's mechanism as first considering the allocation rule that maximizes social welfare with 
respect to the profile of virtual values $\phi_i(\vali)$.  However, if the virtual valuation function is non-monotone, this 
allocation rule will also be non-monotone and hence not incentive compatible.  
The mechanism addresses this issue by ironing the virtual valuation function, which effectively monotonizes the allocation rule.

The motivation for ironing in our construction is quite similar, in that we are given a non-monotone allocation rule that we wish to make
incentive compatible.  Furthermore, we address the issue in a similar way: by ironing the offending non-monotone curve.  One might
therefore suspect that these two monotonization procedures are, in fact, equivalent when restricted to the allocation rule that maximizes virtual welfare.  However, as we will now show, this is not the
case: the mechanisms that result from ironing the virtual valuation function and from ironing the allocation rule are distinct.  Thus, our 
construction does differ, in an essential way, from that of Myerson.

%\textbf{BJL - stuff}
%
%The ironing procedure described in this paper is reminiscent of the
%ironing procedure used by Myerson for maximizing revenue
%\cite{mye-81}.  Myerson's optimal mechanism allocates to maximize
%virtual value (defined by $\phi_i(\vali) = \vali - \frac{1-\disti(\vali)}{\densi(\vali)}$ for agent $i$).
%When the virtual valuation functions are non-monotone, this mechanism
%first irons the virtual valuation functions, then allocates to
%maximize ironed virtual value.  A natural question is whether this is
%identical to ironing the non-monotone allocation rule that would
%result from maximizing non-ironed virtual values.  We now give an example demonstrating that these two mechanisms are distinct.

Let us provide an example to illustrate this difference.
Consider an auction of a single indivisible item to multiple bidders with values drawn i.i.d.~from distribution $F$.
Consider the following distribution $F$: with probability $1/2$, the value is drawn uniformly from $[\frac{3}{8},\frac{1}{2}]$; otherwise, it is drawn uniformly from $(\frac{1}{2},1]$ (see Figure \ref{fig:appMyerson}). The virtual valuation function and ironed virtual valuation function corresponding to this distribution are
$$
\phi(v) = \begin{cases}
2v - \frac{5}{8} & v \in \left[\frac{3}{8},\frac{1}{2}\right] \\
2v - 1 & v \in \left(\frac{1}{2},1\right].
\end{cases}
\quad\quad
\overline{\phi}(v) = \begin{cases}
2v - \frac{5}{8} & v \in \left[\frac{3}{8},\frac{13}{32}\right] \\
\frac{3}{16} & v \in \left(\frac{13}{32},\frac{19}{32}\right] \\
2v - 1 & v \in \left(\frac{19}{32},1\right].
\end{cases}
$$ 

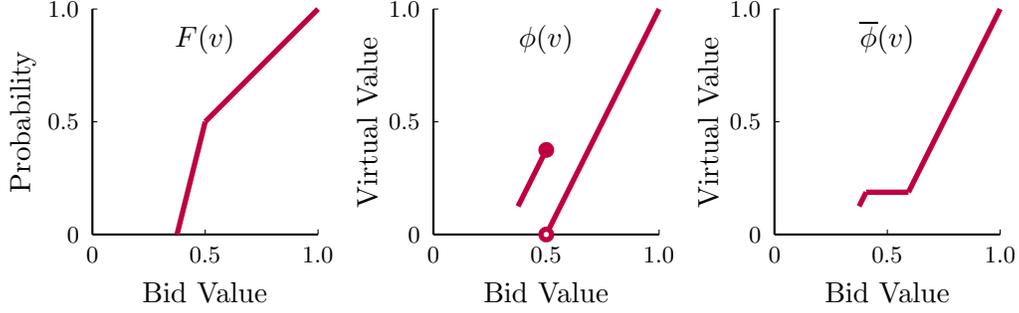
\begin{figure}[t]
\begin{center}
\begin{tabular}{ccc}

{
\def\pshlabel#1{\footnotesize  #1}
\def\psvlabel#1{\footnotesize #1}
\psset{yunit=3cm,xunit=3cm}
\begin{pspicture}(-0.2,-0.3)(1,1)
\psaxes[ticksize=0.4ex,Dx = 0.5,Dy = 0.5]{-}(0,0)(0,0)(1,1)

{%
\psset{linewidth=2pt,linecolor=purple}
\psline{-}(0.375,0)(0.5,0.5)
\psline{-}(0.5,0.5)(1,1)
}

\rput[b](0.5,0.8){$\dist(\val)$}
\rput[b](0.5,-0.3){Bid Value}
\rput[b]{90}(-0.25,0.5){Probability}

\end{pspicture}
}

&

{
\def\pshlabel#1{\footnotesize  #1}
\def\psvlabel#1{\footnotesize #1}
\psset{yunit=3cm,xunit=3cm}
\begin{pspicture}(-0.2,-0.3)(1,1)
\psaxes[ticksize=0.4ex,Dx = 0.5,Dy = 0.5]{-}(0,0)(0,0)(1,1)

{%
\psset{linewidth=2pt,linecolor=purple}
\psline{-*}(0.375,0.125)(0.5,0.375)
\psline{o-}(0.5,0)(1,1)
}

\rput[b](0.5,0.8){$\phi(\val)$}
\rput[b](0.5,-0.3){Bid Value}
\rput[b]{90}(-0.25,0.5){Virtual Value}

\end{pspicture}
}

&

{
\def\pshlabel#1{\footnotesize  #1}
\def\psvlabel#1{\footnotesize #1}
\psset{yunit=3cm,xunit=3cm}
\begin{pspicture}(-0.2,-0.3)(1,1)
\psaxes[ticksize=0.4ex,Dx = 0.5,Dy = 0.5]{-}(0,0)(0,0)(1,1)

{%
\psset{linewidth=2pt,linecolor=purple}
\psline{-}(0.375,0.125)(0.40625,0.1875)
\psline{-}(0.40625,0.1875)(0.59375,0.1875)
\psline{-}(0.59375,0.1875)(1,1)
}

\rput[b](0.5,0.8){$\overline{\phi}(\val)$}
\rput[b](0.5,-0.3){Bid Value}
\rput[b]{90}(-0.25,0.5){Virtual Value}

\end{pspicture}
} 

\end{tabular}
\end{center}
\caption{The distribution $\dist(\val)$ used in Appendix \ref{app:Myerson}, with virtual valuation function $\phi(\val)$ and ironed virtual valuation function  $\overline{\phi}(\val)$.}
\label{fig:appMyerson}
\end{figure}

Suppose $\alg$ is the allocation rule that assigns the item to the agent with highest virtual value.
We now consider the two incentive compatible variants of $\alg$ that we wish to compare.  Namely, 
let $\alg'$ be Myerson's algorithm, which assigns the item to the agent with the highest \emph{ironed} virtual value, and
let $\ialg$ be the ironed algorithm corresponding to $\alg$ (as in Section \ref{sec:ideal}).  
Let $\alloc(\cdot)$, $\alloc'(\cdot)$, and $\ialloc(\cdot)$ denote the allocation curves corresponding to $\alg$, $\alg'$, and $\ialg$, respectively\footnote{We drop the usual subscript of agent index since, by symmetry, the allocation curves are the same for each player.}.  
Our goal is to show that $\ialloc(\cdot) \neq \alloc'(\cdot)$.

%The ironed virtual valuation function is given by:
%$$
%\phi'(v) = \begin{cases}
%2v - 12 & v \in [10,10.25] \\
%8.5 & v \in (10.25,11.75] \\
%2v - 15 & v \in (11.75,15].
%\end{cases}
%$$ 
We observe that the function $\overline{\phi}(\cdot)$ achieves a strict minimum, over its effective range $[\frac{3}{8},1]$, at the point $v = \frac{3}{8}$.  This implies that $\alloc'(\frac{3}{8}) = 0$, since an agent that declares the minimal value can be awarded an allocation only in the $0$-probability event that all other agents report this same value.

On the other hand, it must be that $\ialloc(\frac{3}{8}) = \expect[v]{\alloc(v) \given v \leq z}$ for some $z \in [\frac{3}{8},1]$.  We claim that this value is strictly positive.  Indeed, $\phi(v) > \phi(w)$ for $v \in [\frac{3}{8},\frac{1}{2}]$ and $w \in (\frac{1}{2},\frac{9}{16})$.  This implies that $\alloc(v) > 0$ for all $v \in [\frac{3}{8},\frac{1}{2})$.  We must therefore have $\ialloc(\frac{3}{8}) = \expect[v]{\alloc(v) \given v \leq z} > 0$.

We conclude $\alloc'(\frac{3}{8}) \neq \ialloc(\frac{3}{8})$, and thus the allocation rules $\alg'$ and $\ialg$ are distinct.

\section{Equilibria of Non-Monotone Algorithms}
\label{app:original}

We have shown how to transform a non-monotone
algorithm into a monotone one to obtain a mechanism that is BIC.  It
is notable that the agents could not do this on their own: there
are non-monotone algorithms that, when coupled with any individually-rational and no-positive-transfer\footnote{No positive
  transfers implies that losers have zero payment.}
payment rule, do not
have any BNE with near the expected welfare as the original algorithm
on the true values.

Choose parameter $X \gg n$.  Consider an auction of a single
indivisible item to $n$ bidders with values drawn i.i.d.~from the
following distribution: with probability $1/n$ the value is $X$; with
the remaining probability it is drawn uniformly from $[0,1]$.
Let $\alg$ allocate to the bidder with the
largest value in $\left[0,\frac{1}{n^2}\right]$, if any and breaking ties
randomly; otherwise it
allocates to the bidder with the largest value.  

Consider the expected welfare of $\alg$.  Since with high probability
an agent has value $X$ and no agent has value $1/n^2$ or below, $\alg
= \Omega(X)$.

Next we show that in any BNE most agents will bid $1/n^2$ and the
expected welfare will be the average value of the agents which is
$O(X/n)$.  Thus, the equiligrium is far from the algorithms Bayesian
performance, i.e., the {\em price of stability} is linear.

Consider any mechanism that pairs $\alg$ with an ex-post IR and
no-positive-transfer payment scheme.  We claim that in any BNE of such
a mechanism, an agent with value greater than $\frac{1}{n^2}$ would
instead report value $\frac{1}{n^2}$.  To see this, consider a BNE and
let $p$ denote the probability that some agent declares value
$\frac{1}{n^2}$.  Suppose that $p < 1-\frac{1}{n}$.  If agent $i$ has
value $\vali \in \left[\frac{1}{8},\frac{3}{8}\right]$ and he does not
bid $\frac{1}{n^2}$, then his probability of allocation is at most
$p(3/4+o(1))$ (since otherwise, with high probability, there will be
an agent with value at least $2\vali$ who could improve his utility by
copying agent $i$'s strategy).  The expected utility of agent $i$ is
therefore at most $p(3\vali/4+o(1))$.  On the other hand, agent $i$
could bid $\frac{1}{n^2}$ for an expected utility of at least
$p\left(\vali-\frac{1}{n^2}\right) > p(\vali/2+o(1))$ (since $\vali
\geq 1/8$).  Thus any agent with a value in
$\left[\frac{1}{8},\frac{3}{8}\right]$ will bid $\frac{1}{n^2}$, so $p
\geq 1-\frac{1}{n}$.  We conclude by noting that if $p \geq
1-\frac{1}{n}$, every player with value above $\frac{1}{n^2}$
maximizes his utility by declaring $\frac{1}{n^2}$.

\section{Failure to Preserve Worst-Case Approximations}
\label{app:Bayesian-approx}

We present an example to demonstrate that ideal ironing does not
preserve worst-case approximation ratios.  Consider an auction of 2
objects to 2 unit-demand bidders, with the goal of optimizing social
welfare.  The private value of agent 1 is drawn uniformly from
$\{1,100\}$, and the private value of agent 2 is drawn uniformly from
$\{10,1000,1001\}$.  Let $\alg$ be an approximation algorithm whose
allocation rule is described in Figure \ref{fig.appendix2}.

\begin{figure}
\begin{center}
\begin{tabular}{c||c|c|c|}
 & 10 & 1000 & 1001 \\
\hline \hline
100 & 1, 0 & 0, 1 & 0, 1 \\
\hline
1 & 0, 1 & 1, 1 & 1, 1 \\
\hline
\end{tabular}
\caption{The allocation rule for algorithm $\alg$.  The vertical axis corresponds to $v_1$, the horizontal to $v_2$, and the table entries are of the form ``$x_1$, $x_2$''.  For example, if $(v_1,v_2) = (100,10)$, then $(x_1,x_2) = (1,0)$.}
\label{fig.appendix2}
\end{center}
\end{figure}

%     +-----+------+------+
% 100 | 1/0 | 0/1  | 0/1  |
%     +-----+------+------+
%   1 | 0/1 | 1/1  | 1/1  |
%     +-----+------+------+
%       10    1000   1001

We note that $\alg$ is a worst-case 11/10 approximation algorithm
(where the optimal solution is to always allocate to both players).
Also, $\alg$ is not BIC for agent 1: $E[x_1(1)] = 2/3$, whereas
$E[x_1(100)] = 1/3$.  The ideal monotonization procedure will draw a
new bid $v_1'$ for agent 1 uniformly from $\{1,100\}$, and run $\alg$
on $(v_1',v_2)$.  Call this new algorithm $\ialg$.

Note that if $v_1 = 100$ and $v_2 = 10$, then with probability 1/2
$\ialg$ will take $v_1' = 1$ and choose allocation $(0,1)$, and with
the remaining probability it will take $v_1' = 100$ and choose
allocation $(1,0)$.  Hence, for this set of input values, the expected
welfare obtained by $\ialg$ is $\frac{100+10}{2} = 55$.  Since $110$
is optimal, $\ialg$ is at best a $2$-approximation algorithm, whereas
$\alg$ is an $11/10$-approximation algorithm.  We conclude that ideal
ironing can cause a significant decrease in worst-case approximation
ratios.

\section{Beyond Social Welfare}
\label{app:makespan}

In the ideal model, our general reduction applies to any
single-parameter optimization problem and converts an algorithm into a
mechanism with at least the same expected social welfare.
Unfortunately, this approach does not preserve other relevant
objective values, even monotone ones such as the makespan.  We
illustrate this deficiency of the approach with an example.

Consider the problem of {\em job scheduling on related machines} where the objective is to minimize the makespan.  Here the machines are agents and
each has a (privately-known) speed.  The time a job takes on a machine
is the product of its length and the machine's speed.  The goal of the
mechanism is to assign a given set of jobs with varying lengths to the
machines so as to minimize the time until all machines have finished
processing their jobs, a.k.a., the makespan.

Consider an instance in which we have 10 unit-length jobs, and 5
machines.  We assume a Bayesian setting, where the speeds of the
machines are probabilistic.  The first 4 machines are identical: they
all have speed 2 with probability 1.  The last machine has either
speed 1 or speed 2, each with probability $1/2$.

Suppose $\alg$ behaves in the following way.  When $s_5 = 2$, it will
choose $x_i = 2$ for all $i$, resulting in a makespan of $1$.  When
$s_5 = 1$, $\alg$ sets $\allocs = (6,1,0,0,3)$, resulting in a
makespan of $3$ (whereas the optimal is 1.5).  Thus the average
expected makespan achieved by $\alg$ is $2$.

We note that, for this algorithm, the allocation curve for machine 5
is not monotone.  Our monotonization procedure will therefore iron the
valuation space of machine 5.  The optimal ironing of this curve will
draw $s_5'$ uniformly from $\{1,2\}$.  This causes 4 equally likely
possibilities, corresponding to $(s_5,s_5') \in \{1,2\}^2$.

If $s_5 = s_5'$ then $\alg$ is proceeding as though no ironing
occurred: if $s_5 = s_5' = 1$ then the the makespan is 1, and if $s_5
= s_5' = 2$ the makespan is 3.  Suppose $s_5 = 2$, $s_5' = 1$.  Then
$\alg$ forms allocation $x$ as though machine 5 has speed 1, though it
actually has speed 2.  Hence we obtain $\allocs = (6,1,0,0,3)$, for a
makespan of $3$.  If, on the other hand, $s_5 = 1, s_5' = 2$ then we
obtain $\allocs = (2,2,2,2,2)$, for a makespan of $2$.

We conclude that the expected makespan of the ironed procedure is
$\frac{1+3+3+2}{4} = 2.25$, which is strictly worse than the expected makespan
obtained by the original algorithm.

\section{Recursive Ironing does not Guarantee ex post IC}
\label{app:IC-dependence}

In order for a mechanism to guarantee ex-post incentive compatibility, it must be that, for all $i \in [n]$, the allocation rule $\alloci$ is monotone \emph{for any choice of $\valsmi$}.  Monotonizing each agent's allocation rule independently is insufficient to obtain this goal.  Indeed, it is easy to construct examples where each each agent's allocation curve is monotone in expectation, but non-monotone for a particular choice of the other agents' bids.

One might imagine the following recursive approach for obtaining ex-post incentive compatibility.  Begin by assuming that each agents' input is drawn from the singleton interval $I_i = [\vali,\vali]$, and let $\textbf{I}$ denote the cube $I_1 \times I_2 \times \dotsm \times I_n$.  Choose some agent $i$ whose allocation curve is not monotone, under the assumption that each other agents' values are drawn from cube $\textbf{I}$, and iron that agent's curve under this assumption.  This ironing process may enlarge the interval from which agent $i$'s value is drawn; update $I_i$ to be this new interval.  Repeat this process, choosing a new agent on each iteration, until all agents' curves are monotone.

%On input $\vals$, choose some agent $i$ and use our ironing procedure to monotonize the allocation curve $\alloci$ given that the other agents bid $\valsmi$.  As a result of this monotonization procedure, the effective bid of agent $i$ is no longer necessarily $\vali$, but rather some $\vali'$ chosen from an interval $I_i$ containing $\vali$.  Now choose some other agent $j \neq i$, and monotonize the allocation curve $\alloc_j$ under the assumption that $\vali$ is drawn from interval $I_i$ and the other values correspond to $\vals$.  The effective bid of agent $j$ is now some $v_j'$ chosen from an interval $I_j$.  We then choose another agent to monotonize, and continue in this way until the allocation rules of all agents are monotone, with respect to the intervals from which the agents' values are drawn.

Unfortunately, as we now demonstrate, the above procedure fails to guarantee ex-post incentive compatibility.  Consider an auction setting with 2 agents, such that either both agents receive an allocation or neither does.  Suppose $\alg$ is the algorithm with allocation rule described in Figure \ref{fig.appendix3}.

\begin{figure}
\begin{center}
\begin{tabular}{c||c|c|c|c|c|c|}
 & 1 & 2 & 3 & 4 & 5 & 6 \\
\hline \hline
2 & 0.20 & 0.60 & 0.60 & 0.20 & 0.20 & 0.60 \\
\hline
1 & 0.80 & 0.20 & 0.82 & 0.22 & 0.84 & 0.24 \\
\hline
\end{tabular}
\caption{The allocation rule for algorithm $\alg$.  The vertical axis corresponds to possible values $v_1$, the horizontal axis corresponds to possible values $v_2$, and the table entries denote the probability of allocation to both agents.  For example, if $(v_1,v_2) = (1,4)$, then $(x_1,x_2) = (0.22,0.22)$.}
\label{fig.appendix3}
\end{center}
\end{figure}

Consider the application of our recursive monotonization technique on this algorithm when $(v_1,v_2) = (1,5)$.  Suppose we choose to monotonize the allocation curve for agent $2$.  Applying our monotonization procedure to the values $(0.8,0.2,0.82,0.22,0.84,0.24)$, we obtain ironed intervals $\{1,2\}$, $\{3,4\}$, and $\{5,6\}$.  The resulting allocation rule is shown in figure \ref{fig.appendix3-2}(a).

\begin{figure}
\begin{center}
\begin{tabular}{cc}
\begin{tabular}{c||c|c|c|c|c|c|}
 & 1 & 2 & 3 & 4 & 5 & 6 \\
\hline \hline
2 & 0.40 & 0.40 & 0.40 & 0.40 & 0.40 & 0.40 \\
\hline
1 & 0.50 & 0.50 & 0.52 & 0.52 & 0.54 & 0.54 \\
\hline
\end{tabular} & 
\begin{tabular}{c||c|c|c|c|c|c|}
 & 1 & 2 & 3 & 4 & 5 & 6 \\
\hline \hline
2 & 0.45 & 0.45 & 0.46 & 0.46 & 0.47 & 0.47 \\
\hline
1 & 0.45 & 0.45 & 0.46 & 0.46 & 0.47 & 0.47 \\
\hline
\end{tabular} 
\\
(a) & (b)
\end{tabular}
\caption{The results of the recursive monotonization procedure for algorithm $\alg$, on input $(1,5)$, (a) after 1 step, (b) after 2 steps.}
\label{fig.appendix3-2}
\end{center}
\end{figure}

We next monotonize curve $x_1$ at the point $v_2 = 5$.  This curve is non-monotone for agent $1$, and the resulting ironed interval is $\{1,2\}$. The resulting allocation rule is shown in figure \ref{fig.appendix3-2}(b).  After this monotonization, the allocation curves for all agents are monotone.  The final expected allocation probability for both players is $0.47$ (the entry at $(v_1,v_2) = (1,5)$).

Next consider the application of this technique when $(v_1,v_2) = (2,5)$.  Monotonizing agent $2$ first, we apply our procedure to values $(0.2, 0.6, 0.6, 0.2, 0.2, 0.6)$ and we obtained ironed interval $\{2,3,4,5\}$.  The resulting allocation rule is shown in figure \ref{fig.appendix3-3}(a).

\begin{figure}
\begin{center}
\begin{tabular}{cc}
\begin{tabular}{c||c|c|c|c|c|c|}
 & 1 & 2 & 3 & 4 & 5 & 6 \\
\hline \hline
2 & 0.20 & 0.40 & 0.40 & 0.40 & 0.40 & 0.60 \\
\hline
1 & 0.80 & 0.52 & 0.52 & 0.52 & 0.52 & 0.24 \\
\hline
\end{tabular} & 
\begin{tabular}{c||c|c|c|c|c|c|}
 & 1 & 2 & 3 & 4 & 5 & 6 \\
\hline \hline
2 & 0.50 & 0.46 & 0.46 & 0.46 & 0.46 & 0.42 \\
\hline
1 & 0.50 & 0.46 & 0.46 & 0.46 & 0.46 & 0.42 \\
\hline
\end{tabular} \\
(a) & (b) \\
\multicolumn{2}{c}{
\begin{tabular}{c||c|c|c|c|c|c|}
 & 1 & 2 & 3 & 4 & 5 & 6 \\
\hline \hline
2 & 0.46 & 0.46 & 0.46 & 0.46 & 0.46 & 0.46 \\
\hline
1 & 0.46 & 0.46 & 0.46 & 0.46 & 0.46 & 0.46 \\
\hline
\end{tabular} 
}\\
\multicolumn{2}{c}{(c)}
\end{tabular}
\caption{The results of the recursive monotonization procedure for algorithm $\alg$, on input $(2,5)$, (a) after 1 step, (b) after 2 steps, (c) after 3 steps.}
\label{fig.appendix3-3}
\end{center}
\end{figure}

We next monotonize curve $x_1$ at the point $v_2 = 5$.  This curve is non-monotone for agent $1$, and the resulting ironed interval is $\{1,2\}$. The resulting allocation rule is shown in figure \ref{fig.appendix3-3}(b).  After this monotonization, the allocation curve for agent $2$ is no longer monotone, so we must iron again over the interval $\{1,2,3,4,5,6\}$.  At this point all agents' allocation curves are monotone.  The final expected allocation probability for both players is $0.46$.

What we have shown is that, given $v_2 = 5$, our recursive monotonization procedure generates the allocation rule $x_1(1) = 0.47 > 0.46 = x_1(2)$ for agent $1$.  This procedure therefore does not result in a monotone allocation rule, and hence does not obtain ex-post incentive compatibility.

%%%%%%%%%%%%%%%%%%%%%%%%%%%%%%

\section{Extension to General Valuations}
\label{app:gen-vals}

We now show how to modify our construction from Section \ref{sec:black-box} to obtain a multiplicative error for many problems of interest, such as downward-closed feasibility settings.  To see how this contrasts with Theorem \ref{thm.main.bic.2}, consider a setting in which the expected valuation of each agent is exponentially small\footnote{Since values are scaled to lie in $[0,1]$, this situation can occur whenever the expected valuations of the agents are bounded, but agents can have exponentially larger values with positive probability.}. In this case, any additive error $\epsilon$ such that $\epsilon^{-1}$ is polynomial will dominate the expected welfare of algorithm $\alg$.  We therefore require a more general theorem in order to obtain meaningful results in this setting.

To this end, let $\mumax = \max_i \expect{\vali}$ be the maximum expected valuation of any agent.  The following is a tightened version of Theorem \ref{thm.main.bic.2} in which the loss in social welfare is scaled by $\mumax$.

\begin{theorem}
\label{thm.main.bic}
In the black-box model and general cost settings, for any $\epsilon > 0$,
a BIC algorithm $\alg'$ can be computed from any Bayesian
algorithm $\alg$.  Its social welfare satisfies 
$\alg' \geq \alg - \epsilon\mumax$, and its
runtime is
polynomial in $n$, $1/\epsilon$, and $\log(1/\mumax)$.
%$\softO(n^{9}\epsilon^{-9}\log^5(\val_{max}/\epsilon\mumax))$.  
%In general feasibility settings, the runtime can be improved to
%$\softO(n^5\epsilon^{-5}\log^3(\val_{max}/\epsilon\mumax))$.
\end{theorem}

For the special case of downward-closed set
systems for feasibility problems, we can assume that $\alg \geq \mumax$, since the trivial
algorithm that simply allocates to the single player with the highest
input value attains this value.  This implies the following corollary.

\begin{cor}
\label{cor.main.bic.downclosed}
In the black-box model and downward-closed feasibility settings, 
for any $\epsilon > 0$, 
a BIC Bayesian $\approxratio(1+\epsilon)$-approximation algorithm
$\alg'$ can be computed from any Bayesian
$\approxratio$-approximation algorithm $\alg$.  Its runtime is
polynomial in $n$, $1/\epsilon$, and $\log(1/\mumax)$.
%polynomial in $n$, $\epsilon^{-1}$, and $\alg^{-1}$.
%$\softO(n^5\epsilon^{-5}\log^3(\val_{max}/\epsilon\mumax))$.
\end{cor}

To prove Theorem \ref{thm.main.bic}, we first consider the following variant of Theorem \ref{thm.main.eps-bne}:

\begin{theorem}
\label{thm.main.eps-bne.2}
In the black-box model and general cost settings, for any $\epsilon > 0$,
an $\epsilon$-BIC algorithm $\alg'$ can be
computed from any Bayesian algorithm $\alg$.  Its social welfare satisfies 
$\alg' \geq \alg - \epsilon\mumax$, and its running time is 
polynomial in $n$, $1/\epsilon$, and $\log(1/\mumax)$.
%$\softO(n^3\epsilon^{-2}\log_{1+\epsilon}(\val_{max}/\epsilon\mumax))$.
\end{theorem}

The proof of Theorem \ref{thm.main.eps-bne.2} follows the proof of Theorem \ref{thm.main.eps-bne} from Section \ref{sec:black-box-bne} almost exactly.
Indeed, the only changes required are to replace instances of the inequality $\expect{\vali} \leq 1$ with $\expect{\vali} \leq \mumax$ throughout, and to alter the definition of the discretization of an algorithm $\alg$, Definition \ref{def:discalg}, so that discretization occurs on the intervals
\[ \intseti = \left\{ \left[0,\eps\mumax\right) \right\} \cup  \left\{\left[\eps\mumax(1+\epsilon)^t,\eps\mumax(1+\epsilon)^{t+1}\right)\right\}_{0 \leq t \leq \log_{1+\eps}(1/\eps\mumax)}.\]
We omit further details of the proof of Theorem \ref{thm.main.eps-bne.2}.

We now turn to proving Theorem \ref{thm.main.bic} from Theorem \ref{thm.main.eps-bne.2}.  Our approach will be the same as the proof of Theorem \ref{thm.main.bic.2} from Theorem \ref{thm.main.eps-bne} in Section \ref{sec:black-box-bic}: we consider a convex combination of the almost-monotone algorithm from Theorem \ref{thm.main.eps-bne.2} with the blatantly monotone stair algorithm.  Recall that in Section \ref{sec:black-box-bic} some care was necessary when finding sets $S_1, \dotsc, S_n$.  This task becomes much more difficult when we wish to keep our error bounded by $\eps\mumax$ (rather than $\eps$).  We describe our approach in the next two subsections: we first give an algorithm for general cost settings, then present an optimization for the special case of feasibility settings.

\subsection{Implementing the Stair Algorithm in General Cost Settings}

In general cost settings, algorithm $\stair{\alg}$ may incur negative
value if, for some $i$, the cost of set $S_i$ is large relative to
$\vali$.  To bound the expected welfare of $\combalg{\alg}$ we must therefore limit the
costs of sets $S_1, \dotsc, S_n$.  Our upper bound on cost will depend
on the following quantity, which relates to the structure of the
piecewise constant intervals for algorithm $\alg$.

\begin{definition} 
\label{def:stairthreshold}
Suppose $\alg$ has piece-wise constant allocation rules, where $\cali_i = \{I_1, I_2, \dotsc \}$ are the constant intervals for agent $i$.  The \emph{stair threshold for agent $i$}, $w_i^\alg$, is defined as $w_i^\alg := \max I_1$ if $|\cali_i|>1$; otherwise $w_i^\alg := \infty$.  That is, $w_i^\alg$ is the upper endpoint of the first valuation space interval for agent $i$, assuming the presence of multiple intervals.
\end{definition}
We can now relate the value of $\combalg{\alg}$ to the cost of sets $S_1, \dotsc, S_n$ and the stair thresholds of algorithm $\alg$.
\begin{lem}
\label{lem.stair.cost}
If there exists $X \geq 0$ such that $\cost(S_i) \leq w_i^\alg + X$ for all $i$, then $\combalg{\alg} \geq \alg - \stairfrac (n\mumax + X)$.
\end{lem}
\begin{proof}
By construction, $\combalg{\alg} = (1-\stairfrac)\alg +
\stairfrac\stair{\alg}$. Recall that $\stair{\alg}$ chooses some $i$
uniformly at random, and then either allocates $S_i$ or $\emptyset$.
Moreover, $\stair{\alg}$ will always allocate $\emptyset$ if $\vali$
is in the first piece of the valuation space; that is, if $\vali <
w_i^\alg$.  We therefore conclude that $\stair{\alg} \geq
\frac{1}{n}\sum_i \min\{w_i^\alg - \cost(S_i),0\} \geq -X$.  Also,
$\alg \leq n\mumax$ trivially.  Thus $\combalg{\alg} \geq \alg - \stairfrac
n\mumax + \stairfrac(-X)) = \alg - \stairfrac (n\mumax + X)$.
\end{proof}

Our goal will be to find sets $S_i$ with $\cost(S_i) \leq w_i^\alg +
n\mumax/\sqrt{\eps}$, then apply Lemma \ref{lem.stair.cost} with $X =
n\mumax/\sqrt{\eps}$.  To find such sets, we will apply the same
sampling techniques used in the construction of $\statalg{\alg}$.
That is, for each $i$ and each piece of the valuation space, we will
run $\alg$ on many sample inputs.  As long as $\alloci(I)$ is not too
small on a given interval $I$, we are very likely to find some valid
allocation that includes agent $i$ during the sampling process!
Moreover, we will show that not all sets discovered in this way can
have high cost, so with high probability we will find a set $S_i$ with
cost at most $n\mumax/\sqrt{\eps}$.  To relate the cost of set $S_i$
with $w_i^\alg$, we will also \emph{iron together all left-most
  intervals for which a low-cost set was not found}.  These
ironed-together intervals will then act like a single piece of
valuation space, which will allow us to relate the cost of any set
$S_i$ we \emph{do} find to the stair threshold for the (modified)
algorithm.

\begin{definition}[$\trimalg{\alg}$]
\label{def:stairalg2}
Given piece-wise constant algorithm $\alg$, the \emph{stair-compatible algorithm for $\alg$}, $\trimalg{\alg}$, is as follows:
\begin{enumerate}
\item For each agent $i$: 
\item \quad Let $\intseti = \{I_1, \dotsc, I_k\}$ be the constant valuation space intervals for agent $i$.
\item \quad For each $I_j \in \intseti$, draw $4\eps^{-2}\log(2n/\eps)$ samples from $\dists$ conditional on $\vali \in I_j$, and run $\alg$ on each of these samples.
\item \quad Let $j_i$ be the minimal index such that, for some sample of interval $I_{j_i}$, $\alg$ allocated a set $T_i$ with $T_i \ni i$ and $\cost(T_i) \leq \min I_{j_i} + n\mumax/\sqrt{\eps}$.  Choose $S_i$ to be any such $T_i$.
\item \quad If no such set $T_i$ was returned for \emph{any} interval, take $j_i = k+1$ and $S_i = \{i\}$.
\item \quad Define $\intseti' = \{ I_1 \cup \dotsc \cup I_{j_i-1}, I_{j_i}, \dotsc, I_k \}$.
\item Run $\resampalg[\intsets']{\alg}$.
\end{enumerate}
Note that, as part of the execution of $\trimalg{\alg}$, a set $S_i \ni i$ will be found for each $i$, which is taken to be any set satisfying the conditions on line 4 for interval $I_{j_i}$ (or $\{i\}$ if no sets were found).
\end{definition}

In summary, $\trimalg{\alg}$ samples each constant interval for agent $i$, searching for an appropriate set $S_i$.  We take $I_{j_i}$ to be the leftmost interval for which such a set $S_i$ was found.  All intervals to the left of $I_{j_i}$ are then ironed together.  Thus, regardless of the sampling outcome, $I_{j_i}$ will be the second valuation space piece for agent $i$ in algorithm $\ialg_{\intsets'}$.  Thus $\cost(S_i) \leq w_i^{\trimalg{\alg}} + n\mumax/\sqrt{\eps}$.

\begin{lem}
\label{lem.findsets-bound}
The stair compatible algorithm $\trimalg{\alg}$ for
$\alg$
(Definition~\ref{def:stairalg2}) and stair thresholds
${\mathbf w}^{\trimalg{\alg}}$ (Definition~\ref{def:stairthreshold}) satisfy $\cost(S_i) \leq w_i^{\trimalg{\alg}} +
n\mumax/\sqrt{\eps}$ for all $i$.
\end{lem}
\begin{proof}
For each $i$, if no
set satisfying the conditions on line 4 of $\trimalg{\alg}$ was found
during the sampling of any interval, then $j_i = k+1$ and all
intervals of $\alg$ are ironed together in $\intsets'$.  In this case
$w_i^{\trimalg{\alg}} = \infty$, so $\cost(S_i) \leq
w_i^{\trimalg{\alg}}$ trivially.  Otherwise, by line 4 of
$\trimalg{\alg}$, $\cost(S_i) \leq \min I_{j_i} +
n\mumax/\sqrt{\eps}$.  However, since all intervals to the left of
$I_{j_i}$ are ironed together in $\intsets'$, $I_{j_i}$ will be the
second piecewise constant interval of $\trimalg{\alg}$, and hence
$\min I_{j_i} = w_i^{\trimalg{\alg}}$.  Thus $\cost(S_i) \leq \min
w_i^{\trimalg{\alg}} + n\mumax/\sqrt{\eps}$ as required.
\end{proof}

\begin{lem}
\label{lem.findsets.2}
$\trimalg{\alg} \geq \alg - 2n\mumax\sqrt{\eps}$.
\end{lem}
\begin{proof}
We claim that, with probability at least $1-\frac{\eps}{2}$, for each
agent $i$, the allocation rules for $\trimalg{\alg}$ and $\alg$ will differ only on values $\vali$ for
which $\alloci(\vali) \leq \sqrt{\eps}+\frac{\eps}{2}$.  Before
proving the claim, let us see how it implies the desired result.  The
claim implies that $\trimalg{\alg} \geq \alg -
(\sqrt{\eps}+\frac{\eps}{2})n\mumax$ with probability $1 -
\frac{\eps}{2}$.  For the remaining probability, we note that
$\trimalg{\alg} \geq \alg - \alg \geq \alg - n\mumax$ trivially.
Thus, over all possible outcomes of sampling, we conclude that
\[ \trimalg{\alg} \geq \alg - \left(\sqrt{\eps}+\frac{\eps}{2}\right)n\mumax - \frac{\eps}{2}n\mumax \geq \alg - 2n\mumax\sqrt{\eps} \]
as required.

Let us now prove the claim.  Choose some agent $i$ and suppose that
$\trimalg{\alg}$ and $\alg$ differ on some interval $I$ with
$\alloci(I) \geq \sqrt{\eps}+\frac{\eps}{2}$.  Let $I$ be the leftmost
such interval. For the remainder of the proof we will say that a set
$T$ has \emph{low cost for $I$} if $\cost(T) \leq \min I +
n\mumax/\sqrt{\eps}$.  Then, by the definition of $\intseti'$, it must be
that no set $T \ni i$ with low cost was found during the sampling of
interval $I$ for agent $i$.  Let us bound the probability of this
event.  
%Consider some agent $i$ and some interval $I \in \cali$.  For
%the remainder of the proof we will say that a set $T$ has \emph{low
%  cost for $I$} if $\cost(T) \leq \min I + n\mumax/\sqrt{\eps}$.
Given $\vals \sim \dists$, let $B(\vals)$ be the event
$[\alloci(\vals) \wedge \sum_i \vali \leq \min I +
  n\mumax/\sqrt{\eps}]$.  If event $B(\vals)$ occurs for some sample $\vals$,
this means that $\alg$ returned some allocation $T \ni i$ and
furthermore $\sum_i \vali \leq \min I + n\mumax/\sqrt{\eps}$.  But
note that this allocation must generate non-negative profit (otherwise
it would never be allocated), and hence $T$ must have low cost for
$I$.  Thus $B(\vals)$ is precisely the event that $\alg$ returns a set
$T \ni i$ with low cost for $I$.

Consider the probability of $B(\vals)$.  By Markov's inequality,
$\prob[\vals]{\sum_{j \neq i} \val_j > n\mumax/\sqrt{\eps}} <
\sqrt{\eps}$.  Thus, since $\vali \geq \min I$ with probability $1$
conditional on $\vali \in I$, $\prob[\vals]{\sum_i \vali > \min I +
  n\mumax/\sqrt{\eps}} < \sqrt{\eps}$.  Also,
$\prob[\vals]{\neg\alloci(\vals)\ |\ \vali \in I} = 1-\alloci(I) \leq
1 - (\sqrt{\eps} + \frac{\eps}{2})$.  The union bound then implies
that $\prob[\vals]{\neg B(\vals)} \leq
1-(\sqrt{\eps}+\frac{\eps}{2})+\sqrt{\eps} = 1 - \frac{\eps}{2}$, so
$\prob[\vals]{B(\vals)} \geq \frac{\eps}{2}$.

By Chernoff-Hoeffding inequality, the probability that event $B$ does
not occur even once during $\eps^{-2}\log(2n/\eps)$ samples is at most
$\frac{\eps}{2n}$.  We conclude that the probability that no set $T
\ni i$ with low cost was found during the sampling of interval $I$ is
at most $\frac{\eps}{2n}$.  This is therefore a bound on the
probability that $\trimalg{\alg}$ and $\alg$ differ for agent $i$
on some interval $I$ with $\alloci(I) \geq
\sqrt{\eps}+\frac{\eps}{2}$.  By the union bound, the probability that
this occurs for \emph{any} agent is at most $\frac{\eps}{2}$, as
required.
%
%We now wish to argue that, with high probability, the allocation rules for $\alg$ and $\trimalg{\eps}{\alg}$ differ for agent $i$ only on intervals $I$ with $\alloci(I) < \sqrt{\eps}+\frac{\eps}{2}$.  For each $i$, if $\alloci(I) < \sqrt{\eps}+\frac{\eps}{2}$ for all intervals $I$, then this is trivially true.  Otherwise, let $I_{j_i}$ denote the leftmost interval for which $\alloci(I) \geq \sqrt{\eps}+\frac{\eps}{2}$.  By the union bound, with probability $1 - \eps/2$, for each $i$, a set $S_i \ni i$ with low cost will be found when sampling interval $I_{j_i}$.  In this case, the behaviour of algorithms $\trimalg{\eps}{\alg}$ and $\alg$ differ only on intervals to the left of $I_{j_i}$, all of which satisfy $\alloci(I) < \sqrt{\eps}+\frac{\eps}{2}$.  Thus, conditioning on an event of probability $1 - \frac{\eps}{2}$, \[\trimalg{\eps}{\alg} \geq \left(1-\left(\sqrt{\eps}+\frac{\eps}{2}\right)\right)\alg \geq \alg - \left(\sqrt{\eps}+\frac{\eps}{2}\right)n\mumax.\]  
%For the remaining probability, $\frac{\eps}{2}$, we note that $\alg \leq n\mumax$ trivially.  We conclude that 
%\[\trimalg{\eps}{\alg} \geq \alg - \frac{\eps}{2}n\mumax - \left(\sqrt{\eps}+\frac{\eps}{2}\right)n\mumax \geq \alg - 2\sqrt{\eps} n\mumax\] 
%as required.  
\end{proof}

We are now ready to describe the algorithm used to prove Theorem \ref{thm.main.bic}.

\begin{definition}[$\corralg{\alg}$] Given an algorithm $\alg$ and any $\eps > 0$, the \emph{monotonization of $\alg$}, $\corralg{\alg}$, is $\combalg{\statalg{\trimalg{\discalg{\alg}}}}$.
%\begin{enumerate}
%\item Construct $\discalg{\alg}$, the discretized version of $\alg$.
%\item Construct $\trimalg{\alg}$, the stair-compatible version of $\discalg{\alg}$.  This generates sets $S_1, \dotsc, S_n$.
%\item Construct $\statalg{\alg}$, the statistically ironed algorithm for $\trimalg{\alg}$.
%\item With probability $\stairfrac = 2k(n-1)\eps$, execute $\stair{\trimalg{\alg}}$ with sets $S_1, \dotsc, S_n$.  Else, execute $\statalg{\alg}$.
%\end{enumerate}
\end{definition}

\begin{lem}
\label{lem.corralg.2}
$\corralg{\alg}$ is BIC, and $\corralg{\alg} \geq \alg - 9kn^2\sqrt{\eps}\mumax$.
\end{lem}
\begin{proof}
%Consider the algorithm $\trimalg_{\eps}$ generated during $\corralg{\eps}{\alg}$; recall that the nature of this algorithm is randomized.  With probability $1$, 
For notational convenience define $\alg' = \trimalg{\discalg{\alg}}$.
Lemma \ref{lem.statalg} implies that $\statalg{\alg'}$ is $\eps$-close to a monotone algorithm, and during the construction of $\alg'$ we find sets $S_1, \dotsc, S_n$ with $S_i \ni i$.  Thus $\combalg{\statalg{\alg'}}$ is well-defined and, by Lemma \ref{lem.step}, is BIC.

We note that costs are not affected by our ironing techniques, and, by Lemma \ref{lem.findsets.2}, 
\begin{align*}
\statalg{\alg'} & \geq \alg' - n\eps\mumax = \trimalg{\discalg{\alg}} - n\eps\mumax \\
& \geq \discalg{\alg} - n\eps\mumax - 2n\mumax\sqrt{\eps} \geq \alg - 5\sqrt{\eps}n\mumax.
\end{align*} 
Also, $\cost(S_i) \leq w_i^{\alg'} + n\mumax/\sqrt{\eps} \leq w_i^{\statalg{\alg'}} + n\mumax/\sqrt{\eps} $ for all $i$ by Lemma \ref{lem.findsets-bound}.  Thus, by Lemma \ref{lem.stair.cost}, 
\begin{align*}
\corralg{\alg} & = \combalg{\statalg{\alg'}} \\
& \geq \statalg{\alg'} - \stairfrac (n\mumax + n\mumax/\sqrt{\eps}) \\
& \geq \alg - 5n\sqrt{\eps}\mumax - (2(k-1)n\eps)2n\mumax/\sqrt{\eps} \\
& \geq \alg - 9kn^2\sqrt{\eps}\mumax.
\end{align*}
%Taking $\eps' = (\eps/7kn^2)^2$, we arrive at the desired result.
\end{proof}

Theorem \ref{thm.main.bic} now follows immediately from Lemma \ref{lem.corralg.2} by considering algorithm $\corralg[\eps']{\alg}$, where $\eps' = (\eps/9kn^2)^2 = \eps^2/81k^2n^4$.  The runtime, which is dominated by sampling, is $O(kn(\eps')^{-2}) = \softO(n^9\eps^{-9}\log^5(\val_{max}/\eps\mumax))$.

 %proceeding separately for feasibility problems and general cost problems.

%\begin{proof}
%Since $\corralg$ behaves as $\alg'$ with probability $(1-\mu)$ and $\stair{\alg'}$ always returns a feasible solution, $\corralg \geq (1-\mu)\alg'$.

%To show $\corralg$ is BIC, choose any agent $i$ and any values $\vali < \vali'$; we will show $\corralloci(\vali) \leq \corralloci(\vali')$.  If $\vali, \vali'$ are in the same piece of the valuation space then $\corralloci(\vali) = \corralloci(\vali')$.  Otherwise, since $\alg'$ is $\epsilon'$-close to monotone $\alg$, it must be that $\alloci'(\vali) \leq \alloci'(\vali')-2\epsilon'$.  Furthermore, if $s$ is the allocation rule for 
%$\stair{\alg'}$, then $s_i(\vali) \leq s_i(\vali')+1/(k-1)n$.  We conclude that 
%\begin{equation*}
%\begin{split}
%\corralloci(\vali) & = (1-\mu)\alloci'(\vali) + \mu s_i(\vali) \\
%& \leq \corralloci(\vali') - 2\epsilon' + \mu/(k-1)n \\
%& \leq \corralloci(\vali')
%\end{split}
%\end{equation*}
%as required, since $\mu \geq 2(k-1)n\epsilon'$.
%\end{proof}

\subsection{Feasibility Settings}

In general feasibility settings, where costs are either $0$ or infinite, the performance of algorithm $\corralg{\alg}$ improves significantly.  Specifically, we can improve Lemma \ref{lem.findsets.2} as follows:

\begin{lem}
\label{lem.findsets.3}
In feasibility settings, $\trimalg{\alg} \geq \alg - \eps n\mumax$.
\end{lem}
\begin{proof}
Consider some agent $i$ and interval $I \in \cali$, and suppose $\alloci(I) \geq \frac{\eps}{2}$.  Consider the probability of finding an allocation with cost at most $\min I + n\mumax/\sqrt{\eps}$ when sampling for this interval.  Since costs are either $0$ or $\infty$, this is precisely the probability of finding an allocation that includes agent $i$, which is $\alloci(I) \geq \frac{\eps}{2}$.  By Chernoff-Hoeffding inequality, the probability that this event does not occur even once in $\eps^{-2}\log(n/2\eps)$ samples is at most $\eps/2n$.  We will therefore successfully find a set $S_i \ni i$ with probability at least $1-\eps/2n$.

For each $i$, let $I_{j_i}$ denote the leftmost interval on which $\alloci(I) \geq \eps$.  By the union bound, with probability $1 - \eps/2$ we will find a set $S_i \ni i$ when sampling interval $I_{j_i}$, for all $i$.  In this case, the behavior of algorithms $\trimalg{\alg}$ and $\alg$ differ only on intervals $I$ to the left of $I_{j_i}$, all of which satisfy $\alloci(I) < \frac{\eps}{2}$.  Thus, conditioning on an event of probability $1 - \frac{\eps}{2}$, $\trimalg{\alg} \geq \alg(1-\frac{\eps}{2}) \geq \alg - n\mumax\eps$.  For the remaining probability, $\frac{\eps}{2}$, we note that $\alg \leq n\mumax$ trivially.  We conclude that $\alg_\eps \geq \alg - \eps n\mumax$ unconditionally.
\end{proof}

Using Lemma \ref{lem.findsets.3} instead of Lemma \ref{lem.findsets.2} in the analysis of $\corralg{\alg}$, we find that the statement of Lemma \ref{lem.corralg.2} improves to show that $\corralg{\alg} > \alg - 8kn^2\eps\mumax$ in feasibility settings.  Thus, in feasibility settings, we can improve the runtime of the algorithm in Theorem \ref{thm.main.bic} by taking $\alg'$ to be $\corralg[\eps']{\alg}$ with $\eps' = 8kn^2\eps$, which has a runtime of $O(kn(\eps')^{-2}) = \softO(n^5\eps^{-5}\log^3(\val_{max}/\eps\mumax))$.

\end{document}